\documentclass[3p,11pt]{elsarticle}

\usepackage{hyperref}
\usepackage{latexsym}
\usepackage{xspace}
\usepackage{amssymb}
\usepackage{amsmath}
\usepackage{stmaryrd}
\usepackage[all]{xy}
\usepackage{calc}
\usepackage{tikz}
\usepackage{flushend}
\usepackage{color}
\usepackage{enumerate}
\usepackage[ruled,vlined,linesnumbered]{algorithm2e}

\newtheorem{theorem}{Theorem}[section]
\newtheorem{lemma}[theorem]{Lemma}
\newtheorem{proposition}[theorem]{Proposition}
\newtheorem{corollary}[theorem]{Corollary}

\newtheorem{Definition}[theorem]{Definition}
\newtheorem{Example}[theorem]{Example}
\newtheorem{Remark}[theorem]{Remark}

\newenvironment{definition}{\begin{Definition}\begin{em}}{\end{em}\end{Definition}}
\newenvironment{example}{\begin{Example}\begin{em}}{\end{em}\end{Example}}

\newproof{proof}{Proof}

\newcommand{\email}[1]{\mbox{Email: \url{#1}}}

\def\tuple#1{\langle#1\rangle}
\def\eqref#1{(\ref{#1})}

\newcommand{\E}{\exists}

\def\false{\mathit{false}}
\def\true{\mathit{true}}

\newcommand{\fALC}{$\mathit{f}\!\mathcal{ALC}$\xspace}

\newcommand{\NN}{\mathbb{N}}

\newcommand{\myend}{\mbox{}\hfill{\footnotesize$\blacksquare$}}

\newcommand{\comment}[1]{}

\newcommand{\fand}{\varotimes}

\newcommand{\fequiv}{\Leftrightarrow}

\newcommand{\bbQ}{\mathbb{Q}}

\newcommand{\counter}{\mathit{counter}}

\newcommand{\Null}{\mathit{null}}

\newcommand{\CompFB}{\mbox{$\mathsf{ComputeFuzzyBisimulation}$}\xspace}

\newcommand{\CompFP}{\mbox{$\mathsf{ComputeFuzzyPartition}$}\xspace}
\newcommand{\CompFPt}{\mbox{$\mathsf{ComputeFuzzyPartitionEfficiently}$}\xspace}
\newcommand{\ConvertFPtoFB}{\mathsf{ConvertFP2FB}}
\newcommand{\AuxConvertFPtoFB}{\mathsf{AuxConvertFP2FB}}

\newcommand{\Initialize}{\mathsf{Initialize}}

\newcommand{\RefineAA}{\mathsf{Refine}_{1a}}
\newcommand{\RefineAB}{\mathsf{Refine}_{1b}}
\newcommand{\RefineB}{\mathsf{Refine}_2}

\newcommand{\SV}{\Sigma_V}
\newcommand{\SE}{\Sigma_E}

\newcommand{\bbP}{\mathbb{P}}

\newcommand{\bbX}{\mathbb{X}}
\newcommand{\bbY}{\mathbb{Y}}
\newcommand{\bbB}{\mathbb{B}}

\newcommand{\pushKey}{\mathit{pushKey}}
\newcommand{\popKey}{\mathit{popKey}}
\newcommand{\maxKey}{\mathit{maxKey}}
\newcommand{\departingPComponentEdge}{\mathit{departingPCEdge}}
\newcommand{\sourcePComponentEdge}{\mathit{sourcePCEdge}}
\newcommand{\PComponentEdge}{\mathit{PComponentEdge}}
\newcommand{\pcEdge}{\mathit{pcEdge}}
\newcommand{\edgeLabel}{\mathit{label}}
\newcommand{\edgeOrigin}{\mathit{origin}}
\newcommand{\edgeDest}{\mathit{destination}}
\newcommand{\Vertex}{\mathit{Vertex}}
\newcommand{\VertexList}{\mathit{VertexList}}
\newcommand{\Edge}{\mathit{Edge}}
\newcommand{\vertexID}{\mathit{id}}
\newcommand{\vertexBlock}{\mathit{scBlock}}
\newcommand{\Block}{\mathit{Block}}
\newcommand{\BlockList}{\mathit{BlockList}}
\newcommand{\departingBlocks}{\mathit{departingSubblocks}}
\newcommand{\comingEdges}{\mathit{comingEdges}}

\newcommand{\vertices}{\mathit{vertices}}
\newcommand{\verticesX}{\mathit{vertices\!\_of\!\_X}}

\newcommand{\compound}{\mathit{compound}}
\newcommand{\smallerBlock}{\mathit{smallerBlock}}
\newcommand{\Size}{\mathit{size}}

\newcommand{\pPartition}{\mathit{pPartition}}

\newcommand{\degree}{\mathit{degree}}

\newcommand{\add}{\mathit{add}}

\newcommand{\Empty}{\mathit{empty}}
\newcommand{\keys}{\mathit{keys}}

\newcommand{\createPComponent}{\mathit{createPComponent}}

\newcommand{\EdgeList}{\mathit{EdgeList}}
\newcommand{\addBlock}{\mathit{addBlock}}

\newcommand{\addPComponent}{\mathit{addPComponent}}
\newcommand{\getVertex}{\mathit{getVertex}}
\newcommand{\addLabel}{\mathit{addLabel}}

\newcommand{\pcEdges}{\mathit{pcEdges}}
\newcommand{\edges}{\mathit{edges}}
\newcommand{\clear}{\mathit{clear}}
\newcommand{\removeBlock}{\mathit{removeBlock}}

\newcommand{\bool}{\mathit{bool}}
\newcommand{\processed}{\mathit{processed}}
\newcommand{\first}{\mathit{first}}

\newcommand{\LabelDegree}{\mathit{LabelDegree}}
\newcommand{\labelDegrees}{\mathit{labelDegrees}}
\newcommand{\labelDegreesIdx}{\mathit{labelDegrees\_idx}}
\newcommand{\allDegrees}{\mathit{allDegrees}}

\newcommand{\ComputePComponentEdges}{\mathsf{ComputePComponentEdges}}
\newcommand{\ComputeSubblocks}{\mathsf{ComputeSubblocks}}
\newcommand{\ClearAuxiliaryInfo}{\mathsf{ClearAuxiliaryInfo}}

\newcommand{\DoSplitting}{\mathsf{DoSplitting}}

\SetKwInput{Input}{Input}
\SetKwInput{Output}{Output}
\SetKwInput{Purpose}{Purpose}
\SetKw{Break}{break}
\SetKw{Continue}{continue}
\SetKw{Return}{return}
\SetKw{New}{new\,}

\newcommand{\subblocks}{\mathit{subblocks}}
\newcommand{\length}{\mathit{length}}

\newcommand{\parent}{\mathit{parent}}
\newcommand{\nekst}{\mathit{next}}
\newcommand{\prev}{\mathit{prev}}

\newcommand{\elements}{\mathit{elements}}
\newcommand{\refineA}{\mathit{refine}_1}
\newcommand{\refineB}{\mathit{refine}_2}

\newcommand{\Vector}{\mathit{Vector}}

\journal{arXiv}

\begin{document}
\sloppy
	
\begin{frontmatter}
	
\title{Computing the Fuzzy Partition Corresponding to the Greatest Fuzzy Auto-Bisimulation of a Fuzzy Graph-Based Structure}

\author{Linh Anh Nguyen}
\ead{nguyen@mimuw.edu.pl}

\address{
	Institute of Informatics, University of Warsaw, 
	Banacha 2, 02-097 Warsaw, Poland 
}

\begin{abstract}
Fuzzy graph-based structures such as fuzzy automata, fuzzy labeled transition systems, fuzzy Kripke models, fuzzy social networks and fuzzy interpretations in fuzzy description logics are useful in various applications. Given two states, two actors or two individuals $x$ and $x'$ in such structures $G$ and $G'$, respectively, the similarity degree between them can be defined to be $Z(x,x')$, where $Z$ is the greatest fuzzy bisimulation between $G$ and $G'$ with respect to some t-norm-based fuzzy logic. Such a similarity measure has the Hennessy-Milner property of fuzzy bisimulations as a strong logical foundation. A fuzzy bisimulation between a fuzzy structure $G$ and itself is called a fuzzy auto-bisimulation of $G$. The greatest fuzzy auto-bisimulation of an image-finite fuzzy graph-based structure is a fuzzy equivalence relation. It is useful for classification and clustering.

In this paper, we design an efficient algorithm with the complexity $O((m\log{l} + n)\log{n})$ for computing the fuzzy partition corresponding to the greatest fuzzy auto-bisimulation of a finite fuzzy labeled graph $G$ under the G\"odel semantics, where $n$, $m$ and $l$ are the number of vertices, the number of non-zero edges and the number of different fuzzy degrees of edges of $G$, respectively. Our notion of fuzzy partition is novel, defined only for finite sets with respect to the G\"odel t-norm, with the aim to facilitate the computation of the greatest fuzzy auto-bisimulation. By using that algorithm, we also provide an algorithm with the complexity $O(m\cdot\log{l}\cdot\log{n} + n^2)$ for computing the greatest fuzzy bisimulation between two finite fuzzy labeled graphs under the G\"odel semantics. This latter algorithm is better (has a lower complexity order) than the previously known algorithms for the considered problem. Our algorithms can be restated for the other mentioned fuzzy graph-based structures.
\end{abstract}


\end{frontmatter}

\section{Introduction}

Labeled transition systems, automata, Kripke models, social networks and interpretations in description logics have in common that they are graph-based structures. The relation $Z$ that specifies whether two states, two actors or two individuals $x$ and $x'$ in such structures $G$ and $G'$, respectively, behave equivalently or are equivalent from the logical point of view has useful applications in practice. For example, when $G'$ is the same as $G$, $Z$ is an equivalence relation and we can use it to minimize $G$ or exploit it in classification or clustering. The relation $Z$ is the largest bisimulation, also called the bisimilarity relation, between $G$ and $G'$. Bisimulation~\cite{vanBenthemCorr,HennessyM85} is a natural notion of equivalence that arose in modal logic and labeled transition systems and has been widely studied for all of the mentioned kinds of graph-based structures.
 
To deal with vagueness and impreciseness, fuzzy graph-based structures are used instead of crisp ones. There are two kinds
of bisimulation, namely {\em crisp} and {\em fuzzy}, between fuzzy graph-based structures. Crisp bisimulations characterize indiscernibility of states/actors/individuals. The Hennessy-Milner property of crisp bisimulations~\cite{ai/FanL14,Fan15,FSS2020} states that, if $Z$ is the largest crisp bisimulation between two image-finite fuzzy graph-based structures $G$ and $G'$, then $Z(x,x')$ holds iff $\varphi^G(x) = \varphi^{G'}(x')$ for all formulas $\varphi$ of a certain fuzzy modal/description logic with the Baaz projection operator or involutive negation, where $\varphi^G(x)$ (respectively, $\varphi^{G'}(x')$) means the degree in which $x$ (respectively, $x'$) has the property $\varphi$ (or is an instance of $\varphi$). 
On the other hand, fuzzy bisimulations characterize similarity between states/actors/individuals. The Hennessy-Milner property of fuzzy bisimulations~\cite{ai/FanL14,Fan15,FSS2020,FBSML} states that, if $Z$ is the greatest fuzzy bisimulation between two image-finite fuzzy graph-based structures $G$ and $G'$, then $Z(x,x') = \inf\{\varphi^G(x) \fequiv \varphi^{G'}(x') \mid$ $\varphi$ is a formula of a certain fuzzy modal/description logic$\}$, where $\fequiv$ is the fuzzy equivalence in the considered logic.
Crisp bisimulations have been defined and studied for fuzzy transition systems~\cite{CaoCK11,CaoSWC13,DBLP:journals/fss/WuD16,DBLP:journals/ijar/WuCHC18,DBLP:journals/fss/WuCBD18}, weighted automata~\cite{DamljanovicCI14}, fuzzy modal logics~\cite{EleftheriouKN12,Fan15,aml/MartiM18,fuin/Diaconescu20} and fuzzy description logics~\cite{FSS2020}. 
Fuzzy bisimulations have been defined and studied for fuzzy automata~\cite{CiricIDB12,CiricIJD12}, weighted/fuzzy social networks~\cite{ai/FanL14,IgnjatovicCS15}, fuzzy modal logics~\cite{Fan15,FBSML} and fuzzy description logics~\cite{FSS2020,minimization-by-fBS,TFS2020}. 

This work concerns computing the greatest fuzzy bisimulation between two finite fuzzy graph-based structures. A discussion on related work, based on~\cite{CompCB-arxiv}, is presented below.

\subsection{Related Work}

In~\cite{CiricIJD12} {\'C}iri{\'c} {\em et at.} gave an algorithm for computing the greatest fuzzy simulation/bisimulation (of any kind defined in~\cite{CiricIDB12}) between two finite fuzzy automata. 
They did not provide a detailed complexity analysis.  
Following~\cite{CiricIJD12}, Ignjatovi{\'c} {\em et at.}~\cite{IgnjatovicCS15} gave an algorithm with the complexity $O(ln^5)$ for computing the greatest fuzzy bisimulation between two fuzzy social networks, where $n$ is the number of nodes in the networks and $l$ is the number of different fuzzy values appearing during the computation. Later Mici{\'c} {\em et at.}~\cite{MicicJS18} provided algorithms with the complexity $O(ln^5)$ for computing the greatest right/left invariant fuzzy quasi-order/equivalence of a finite fuzzy automaton, where $n$ is the number of states of the considered automaton and $l$ is the number of different fuzzy values appearing during the computation. These relations are closely related to the fuzzy simulations/bisimulations studied in~\cite{CiricIDB12,CiricIJD12}. 
In~\cite{TFS2020} Nguyen and Tran provided an algorithm with the complexity $O((m+n)n)$ for computing the greatest fuzzy bisimulation between two finite fuzzy interpretations in the fuzzy description logic \fALC under the G\"odel semantics, where $n$ is the number of individuals and $m$ is the number of non-zero instances of roles in the given fuzzy interpretations. They also adapted that algorithm for computing fuzzy simulations/bisimulations between finite fuzzy automata and obtained algorithms with the same complexity order. 

In~\cite{DBLP:journals/fss/WuCBD18} Wu {\em et al.} studied algorithmic and logical characterizations of crisp bisimulations for nondeterministic fuzzy transition systems (NFTSs)~\cite{CaoSWC13}. They gave an algorithm with the complexity $O(m^2n^4)$ for testing crisp bisimulation (i.e., for checking whether two given states are bisimilar), where $n$ is the number of states and $m$ is the number of transitions in the underlying NFTS. 
In~\cite{StanimirovicSC2019} Stanimirovi{\'c} {\em et at.} provided algorithms with the complexity $O(n^3)$ for computing the greatest right/left invariant Boolean (crisp) equivalence matrix of a weighted automaton over an additively idempotent semiring. Such matrices are closely related to crisp bisimulations. 
In~\cite{CompCB-arxiv} Nguyen and Tran gave an algorithm with the complexity $O((m\log{l} + n)\log{n})$ for computing the (crisp) partition corresponding to the largest crisp bisimulation of a given finite fuzzy labeled graph, where $n$, $m$ and $l$ are the number of vertices, the number of nonzero edges and the number of different fuzzy degrees of edges of the input graph, respectively. They also studied a similar problem for the setting with counting successors, which corresponds to the case with qualified number restrictions in description logics and graded modalities in modal logics. In particular, they provided an algorithm with the complexity $O((m\log{m} + n)\log{n})$ for the considered problem in that setting. 

As the background, also recall that Hopcroft~\cite{Hopcroft71} gave an efficient algorithm with the complexity $O(n\log{n})$ for minimizing states in a (crisp) deterministic finite automaton, and Paige and Tarjan~\cite{PaigeT87} gave efficient algorithms with the complexity $O((m+n)\log{n})$ for computing the coarsest partition of a finite (crisp) graph, for both the settings with stability or size-stability. As mentioned in~\cite{PaigeT87}, an algorithm with the same complexity order for the second setting was given earlier by Cardon and Crochemore~\cite{DBLP:journals/tcs/CardonC82}. 

\subsection{Motivation and Our Contributions}

As discussed in the previous subsection, before the current work, the best known algorithm for computing the greatest fuzzy bisimulation between two finite fuzzy graph-based structures under the G\"odel semantics was given by Nguyen and Tran~\cite{TFS2020} and has the complexity $O((m+n)n)$. The motivation of the current work is to develop a more efficient algorithm for the same problem. 

In this article, by exploiting the ideas and techniques of the works~\cite{Hopcroft71,PaigeT87,CompCB-arxiv,TFS2020}, we develop an efficient algorithm with the complexity $O((m\log{l} + n)\log{n})$ for computing the fuzzy partition corresponding to the greatest fuzzy auto-bisimulation of a finite fuzzy labeled graph $G$ under the G\"odel semantics, where $n$, $m$ and $l$ are the number of vertices, the number of non-zero edges and the number of different fuzzy degrees of edges of the input graph $G$, respectively. Our notion of fuzzy partition is novel, defined only for finite sets with respect to the G\"odel t-norm, with the aim to facilitate the computation of the greatest fuzzy auto-bisimulation. By using that algorithm, we also provide an algorithm with the complexity $O(m\!\cdot\!\log{l}\!\cdot\!\log{n} + n^2)$ for computing the greatest fuzzy bisimulation between two finite fuzzy labeled graphs under the G\"odel semantics. Taking $l = n^2$ for the worst case, the latter complexity order can be simplified to $O(m\log^2(n) + n^2)$. This latter algorithm is better (has a lower complexity order) than the previously known algorithms for the considered problem. 

Our algorithms can be restated for other fuzzy graph-based structures such as fuzzy automata, fuzzy labeled transition systems, fuzzy Kripke models, fuzzy social networks and fuzzy interpretations in fuzzy description logics. 

\subsection{The Structure of This Work}

The rest of this work is structured as follows. In Section~\ref{section: prel}, we give preliminaries on fuzzy sets, fuzzy labeled graphs and fuzzy bisimulations. Section~\ref{sec: FP} is devoted to fuzzy partitions. In Section~\ref{sec: skeleton alg}, we present the skeleton of our algorithm for computing the fuzzy partition corresponding to the greatest fuzzy auto-bisimulation of a finite fuzzy labeled graph under the G\"odel semantics and prove its correctness. In Section~\ref{sec: impl}, we give details on how to implement that algorithm so that its complexity is of order $O((m\log{l} + n)\log{n})$. 
In Section~\ref{sec: comp FB}, we use that improved algorithm to design an algorithm with the complexity $O(m\!\cdot\!\log{l}\!\cdot\!\log{n} + n^2)$ for computing the greatest fuzzy bisimulation between two finite fuzzy labeled graphs under the G\"odel semantics. Section~\ref{sec: conc} contains conclusions. 

\section{Preliminaries}
\label{section: prel}

Recall that a {\em crisp partition} of a non-empty set $X$ is a set of pairwise disjoint non-empty subsets of $X$ whose union is equal to~$X$. Given a crisp partition $\bbP$, by a {\em component} of $\bbP$ we mean an element of the set $\bbP$ (we reserve the term ``block'' for another meaning). 
Given an equivalence relation $\sim$ on $X$, the crisp partition corresponding to $\sim$ is \mbox{$\{[x]_\sim \mid x \in X\}$}, where $[x]_\sim$ is the equivalence class of $x$ w.r.t.~$\sim$ (i.e., \mbox{$[x]_\sim = \{x' \in X \mid x' \sim x\}$}). 

Given crisp partitions $\bbP$ and $\bbQ$ of $X$, we say that $\bbP$ is a {\em refinement} of $\bbQ$ if, for every $Y_1 \in \bbP$, there exists $Y_2 \in \bbQ$ such that $Y_1 \subseteq Y_2$. In that case we also say that $\bbQ$ is {\em coarser} than $\bbP$. By this definition, every crisp partition is coarser than itself.

\subsection{Fuzzy Sets and Operators}

We use two fuzzy operators of the G\"odel family, which are defined as follows for $x, y \in [0,1]$:
\begin{eqnarray*}
	x \fand y & = & \min\{x,y\} \\
	(x \fequiv y) & = & (\textrm{if $x = y$ then 1 else $\min\{x,y\}$}). 
\end{eqnarray*}

Given a set $X$, a function $f: X \to [0,1]$ is called a {\em fuzzy subset} of $X$. 
If $f$ is a fuzzy subset of $X$ and $x \in X$, then $f(x)$ means the fuzzy degree in which $x$ belongs to the subset. 
For $\{x_1,\ldots,x_n\} \subseteq X$ and $\{a_1,\ldots,a_n\} \subset [0,1]$, we write $\{x_1:a_1, \ldots, x_n:a_n\}$ to denote the fuzzy subset $f$ of $X$ such that $f(x_i) = a_i$ for $1 \leq i \leq n$ and $f(x) = 0$ for $x \in X \setminus \{x_1,\ldots,x_n\}$. 

If $f$ and $g$ are fuzzy subsets of $X$, then we write $f \leq g$ to denote that $f(x) \leq g(x)$ for all $x \in X$. If $f \leq g$, then we say that $g$ is {\em greater than or equal to} $f$. If $F$ is a set of fuzzy subsets of $X$, then by $\sup F$ we denote the fuzzy subset of $X$ specified by $(\sup F)(x) = \sup \{f(x) \mid f \in F\}$. 
As usual, if $f \in F$ and $f = \sup F$, then $f$ is called the {\em greatest} element of $F$.

Let $X$, $Y$ and $Z$ be non-empty sets. A fuzzy subset of $X \times Y$ is called a {\em fuzzy relation} between $X$ and $Y$. A fuzzy relation between $X$ and itself is called a fuzzy relation on~$X$. Given fuzzy relations $f: X \times Y \to [0,1]$ and $g: Y \times Z \to [0,1]$, the {\em composition} of $f$ and $g$, denoted by $f \circ g$, is the fuzzy relation between $X$ and $Z$ such that, for every $x \in X$ and $z \in Z$,  
\[ 
(f \circ g)(x,z) = \sup\{f(x,y) \fand g(y,z) \mid y \in Y\}.
\]
The converse \mbox{$f^- : Y \times X \to [0,1]$} of $f$ is defined by $f^-(y,x) = f(x,y)$.

A fuzzy relation $f: X \times X \to [0,1]$ is 
\begin{itemize}
	\item {\em reflexive} if $f(x,x) = 1$ for all $x \in X$, 
	\item {\em symmetric} if $f = f^-$,
	\item {\em transitive} if $f \circ f \leq f$. 
\end{itemize}
It is a {\em fuzzy equivalence relation} if it is reflexive, symmetric and transitive. 

\subsection{Fuzzy Bisimulations}

A {\em fuzzy labeled graph}, hereafter called a {\em fuzzy graph} for short, is a structure $G = \tuple{V, E, L, \SV, \SE}$, where $V$ is a non-empty set of vertices, $\SV$ (respectively, $\SE$) is a set of vertex labels (respectively, edge labels), \mbox{$E: V \times \SE \times V \to [0,1]$} is called the fuzzy set of labeled edges, and $L: V \to (\SV \to [0,1])$ is called the labeling function of vertices. 
Given vertices $x,y \in V$, a vertex label $p \in \SV$ and an edge label $r \in \SE$, $L(x)(p)$ means the degree in which $p$ is a member of the label of~$x$, and $E(x,r,y)$ means the degree in which there is an edge from $x$ to $y$ labeled by~$r$. 
The graph $G$ is {\em finite} if all the sets $V$, $\SV$ and $\SE$ are finite. It is {\em image-finite} if the set $\{y \mid E(x,r,y) > 0\}$ is finite for all $x \in V$ and $r \in \SE$. 

Fuzzy graphs are used as fuzzy labeled transition systems (FLTSs), fuzzy automata, fuzzy Kripke models and fuzzy interpretations in fuzzy description logics. For example, in the terminology of FLTSs, vertices, edges, edge labels and vertex labels represent states, transitions, actions and atomic properties of states, respectively. Recall that fuzzy bisimulations have been defined and studied for fuzzy automata~\cite{CiricIDB12,CiricIJD12}, weighted/fuzzy social networks~\cite{ai/FanL14,IgnjatovicCS15}, fuzzy modal logics~\cite{Fan15,FBSML} and fuzzy description logics~\cite{FSS2020,minimization-by-fBS,TFS2020}. 
We give below their definition, which is based on \cite{FSS2020,FBSML} and equivalent to the one in \cite{Fan15} when $|\SE| = 1$ and the graphs are image-finite. 

\begin{definition}
Let $G = \tuple{V, E, L, \SV, \SE}$ and $G' = \tuple{V', E', L', \SV, \SE}$ be fuzzy graphs over the same signature $\tuple{\SV, \SE}$. A fuzzy relation $Z \subseteq V \times V' \to [0,1]$ is called a {\em fuzzy bisimulation} between $G$ and $G'$ if the following conditions hold for all $p \in \SV$, $r \in \SE$ and all possible values for the free variables:
\begin{eqnarray}
\!\!\!\!\!\!\!\!\!\!&& Z(x,x') \leq (L(x)(p) \fequiv L'(x')(p)) \label{eq: FB1 a} \\
\!\!\!\!\!\!\!\!\!\!&& \E y' \in V'\ (Z(x,x') \fand E(x,r,y) \leq E'(x',r,y') \fand Z(y,y')) \label{eq: FB2 a} \\
\!\!\!\!\!\!\!\!\!\!&& \E y \in V\ (Z(x,x') \fand E'(x',r,y') \leq E(x,r,y) \fand Z(y,y')). \label{eq: FB3 a}
\end{eqnarray}
\end{definition}

\begin{example}\label{example: JHFJH}
Let $\SV = \{p\}$, $\SE = \{r\}$ and let $G = \tuple{V, E, L, \SV, \SE}$ and $G' = \tuple{V', E', L', \SV, \SE}$ be the fuzzy graphs depicted and specified as follows:
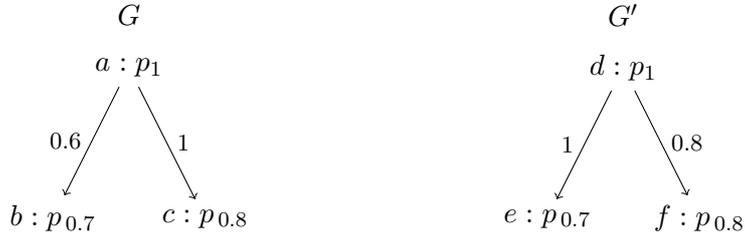
\begin{center}
	\begin{tikzpicture}
	\node (x0) {};
	\node (x) [node distance=1.5cm, right of=x0] {};
	\node (I) [node distance=0.0cm, below of=x] {$G$};
	\node (I1) [node distance=6.5cm, right of=I] {$G'$};
	\node (u) [node distance=0.7cm, below of=I] {$a:p_1$};
	\node (ub) [node distance=2.0cm, below of=u] {};
	\node (v) [node distance=1.0cm, left of=ub] {$b:p_{\,0.7}$};
	\node (w) [node distance=1.0cm, right of=ub] {$c:p_{\,0.8}$};
	\node (up) [node distance=0.7cm, below of=I1] {$d:p_1$};
	\node (ubp) [node distance=2.0cm, below of=up] {};
	\node (vp) [node distance=1.0cm, left of=ubp] {$e:p_{\,0.7}$};
	\node (wp) [node distance=1.0cm, right of=ubp] {$f:p_{\,0.8}$};
	\draw[->] (u) to node [left]{\footnotesize{0.6}} (v);
	\draw[->] (u) to node [right]{\footnotesize{1}} (w);
	\draw[->] (up) to node [left]{\footnotesize{1}} (vp);
	\draw[->] (up) to node [right]{\footnotesize{0.8}} (wp);
	\end{tikzpicture}
\end{center}	
\begin{itemize}
\item $V = \{a,b,c\}$, $E = \{\tuple{a,r,b}\!:\!0.6, \tuple{a,r,c}\!:\!1\}$, $L(a)(p) = 1$, $L(b)(p) = 0.7$, $L(c)(p) = 0.8$;
\item $V' = \{d,e,f\}$, $E' = \{\tuple{d,r,e}\!:\!1, \tuple{d,r,f}\!:\!0.8\}$, $L'(d)(p) = 1$, $L'(e)(p) = 0.7$, $L'(f)(p) = 0.8$.
\end{itemize}
It can be checked that $Z = \{\tuple{a,d}\!:\!0.7$, $\tuple{b,e}\!:\!1$, $\tuple{b,f}\!:\!0.7$, $\tuple{c,e}\!:\!0.7$, $\tuple{c,f}\!:\!1\}$ is the greatest bisimulation between $G$ and $G'$. 
\myend
\end{example}

\begin{definition}
Given a fuzzy graph $G = \tuple{V, E, L, \SV, \SE}$, a fuzzy relation $Z \subseteq V \times V \to [0,1]$ is called a {\em fuzzy auto-bisimulation of $G$}, or a {\em fuzzy bisimulation of $G$} for short, if it is a fuzzy bisimulation between $G$ and itself, i.e., if the following conditions hold for all $p \in \SV$, $r \in \SE$ and all possible values for the free variables:
\begin{eqnarray}
\!\!\!\!\!\!\!\!\!\!&& Z(x,x') \leq (L(x)(p) \fequiv L(x')(p)) \label{eq: FB1} \\
\!\!\!\!\!\!\!\!\!\!&& \E y' \in V\ (Z(x,x') \fand E(x,r,y) \leq E(x',r,y') \fand Z(y,y')) \label{eq: FB2} \\
\!\!\!\!\!\!\!\!\!\!&& \E y \in V\ (Z(x,x') \fand E(x',r,y') \leq E(x,r,y) \fand Z(y,y')). \label{eq: FB3}
\end{eqnarray}
\end{definition}

It is known that the greatest fuzzy bisimulation of every image-finite fuzzy graph exists and is a fuzzy equivalence relation (see, e.g., \cite{CiricIDB12,minimization-by-fBS,FBSML}). 

Let $G = \tuple{V, E, L, \SV, \SE}$ and $G' = \tuple{V', E', L', \SV, \SE}$ be fuzzy graphs over the same signature $\tuple{\SV, \SE}$. The {\em disjoint union} of $G$ and $G'$, denoted by $G \uplus G'$, is the fuzzy graph $G'' = \tuple{V'', E'', L'', \SV, \SE}$ such that $V'' = V \uplus V'$, $E'' = E \uplus E'$ and $L'' = L \uplus L'$. 

\begin{proposition}
Let $G = \tuple{V, E, L, \SV, \SE}$ and $G' = \tuple{V', E', L', \SV, \SE}$ be image-finite fuzzy graphs over the same signature $\tuple{\SV, \SE}$ and let $G'' = G \uplus G'$. Let $Z$ be the greatest fuzzy bisimulation of $G''$. Then, $Z|_{V \times V'}$ is the greatest fuzzy bisimulation between $G$ and~$G'$.
\end{proposition}

This proposition follows from the Hennessy-Milner property of fuzzy bisimulations~\cite[Theorem~3.7]{FSS2020}. 

\section{Fuzzy Partitions}
\label{sec: FP}

Fuzzy partitions have been studied by a considerable number of authors (see, e.g., \cite{OVCHINNIKOV1991107,DBLP:conf/ismvl/Schmechel95,DBLP:journals/isci/BaetsCK98,DBLP:journals/fss/CiricIB07} and references therein). In these cited papers, a fuzzy partition is defined to be the set of all fuzzy equivalence classes of some fuzzy equivalence relation. The notion defined in this way has many interesting properties. In particular, it can be defined for an infinite set over a complete residuated lattice. 
In this section, we introduce and study a novel notion of fuzzy partition, which is defined only for finite sets with respect to the G\"odel t-norm, with the aim to facilitate the computation of the greatest fuzzy auto-bisimulation of a finite fuzzy labeled graph under the G\"odel semantics.  

In this section, let $X$ be a finite set. To represent a collection of elements from $X$ we use an abstract class (type) called {\em block} with two subclasses, which are named {\em fuzzy block} and {\em crisp block}. The intuition behind these kinds of blocks is as follows.
\begin{itemize}
	\item A crisp block is a collection of elements from $X$, whereas a fuzzy block is a collection of blocks called {\em subblocks}. 
	If $B$ is a crisp block, then we define $B.\elements()$ to be the set of elements of $B$. 
	If $B$ is a fuzzy block, then we define $B.\subblocks()$ to be the set of subblocks of $B$, and inductively define $B.\elements()$ to be $\bigcup \{B'.\elements() \mid B' \in B.\subblocks()\}$. In other words, $B.\elements()$ is the set of all elements belonging the collection represented by~$B$. 
	
	\item A fuzzy block $B$ has the attribute $B.\degree \in [0,1)$ with the meaning that, treating $B$ as a fuzzy equivalence class of a fuzzy equivalence relation $f$, for every $x,y \in B.\elements()$, $f(x,y) \geq B.\degree$, and furthermore, if $B_1$ and $B_2$ are different subblocks of $B$, $x \in B_1.\elements()$ and $y \in B_2.\elements()$, then $f(x,y) = B.\degree$. 
	
	\item A crisp block $B$ has the attribute $B.\degree = 1$. 
\end{itemize}
We assume the following restrictions:
\begin{itemize}
	\item if $B$ is a crisp block, then $B.\elements() \neq \emptyset$;
	\item if $B$ is a fuzzy block, then $|B.\subblocks()| > 1$;
	\item if a fuzzy block $B$ is a subblock of a fuzzy block $B'$, then $B.\degree > B'.\degree$;
	\item if $B_1$ and $B_2$ are different subblocks of a block $B'$, then the sets $B_1.\elements()$ and $B_2.\elements()$ are disjoint. 
\end{itemize}
We also assume that blocks $B_1$ and $B_2$ are {\em equal} iff:
\begin{itemize}
	\item either both $B_1$ and $B_2$ are crisp and $B_1.\elements() = B_2.\elements()$,
	\item or both $B_1$ and $B_2$ are fuzzy, $B_1.\degree = B_2.\degree$ and $B_1.\subblocks() = B_2.\subblocks()$ (in the sense that each block from $B_1.\subblocks()$ is equal to some block from $B_2.\subblocks()$ and vice versa).
\end{itemize}

A crisp block $B$ with $B.\elements() = \{a_1,\ldots,a_k\}$ is denoted by $\{a_1,\ldots,a_k\}_1$. 
A fuzzy block $B$ with $B.\degree = d$ and $B.\subblocks() = \{B_1,\ldots,B_k\}$ is denoted by $\{Z_1,\ldots,Z_k\}_d$, where $Z_i$ is the denotation of $B_i$, for $1 \leq i \leq k$. 

A block $B$ is called a {\em fuzzy partition} of $X$ if $B.\elements() = X$. 

\begin{definition}\label{def: HFJAA}
	Given a fuzzy equivalence relation $f: X \times X \to [0,1]$, the {\em fuzzy partition corresponding to $f$} is the block $B$ defined inductively as follows:
	\begin{enumerate}
		\item if $f(x,x') = 1$ for all $x,x' \in X$, then $B$ is the crisp block such that $B.\elements() = X$;
		\item else:
		\begin{enumerate}
			\item let $d = \min\{f(x,x') \mid x,x' \in X\}$;
			\item\label{step: HFJAA 2b} let $\sim$ be the (crisp) equivalence relation on $X$ such that $x \sim x'$ iff $f(x,x') > d$;
			\item let $X_1,\ldots,X_n$ be all the (crisp) equivalence classes of $\sim$;
			\item for $1 \leq i \leq n$, let $f_i$ be the restriction of $f$ to $X_i \times X_i$ and let $B_i$ be the fuzzy partition corresponding to $f_i$;
			\item $B$ is the fuzzy block such that $B.\degree = d$ and $B.\subblocks() = \{B_1,\ldots,B_n\}$.
			\myend
		\end{enumerate}
	\end{enumerate}
\end{definition}

In the above definition, since $f$ is a fuzzy equivalence relation, the relation $\sim$ defined as $\{\tuple{x,x'} \in X \times X \mid f(x,x') > d\}$ is really a (crisp) equivalence relation. This guarantees that the fuzzy partition corresponding to a fuzzy equivalence relation is well defined. 
Here, we implicitly use the assumption that $X$ is finite and $\fand$ is the G\"odel t-norm. 

\begin{definition}
	Let $B$ be a fuzzy partition of $X$. The {\em fuzzy equivalence relation corresponding to $B$} is the fuzzy relation $f: X \times X \to [0,1]$ defined inductively as follows:
	\begin{itemize}
		\item if $B$ is a crisp block, then $f(x,x') = 1$ for all $x,x' \in X$;
		\item else:
		\begin{itemize}
			\item let $d = B.\degree$ and $\{B_1,\ldots,B_n\} = B.\subblocks()$;
			\item for $1 \leq i \leq n$, let $X_i = B_i.\elements()$ and let $f_i$ be the fuzzy equivalence relation corresponding to $B_i$;
			\item for $x,x' \in X$, $f(x,x') =$ (if $\{x,x'\} \subseteq X_i$ for some $1 \leq i \leq n$ then $f_i(x,x')$ else $d$).
			\myend
		\end{itemize}
	\end{itemize}  
\end{definition}

It is easy to check that the fuzzy relation $f$ defined above is really a fuzzy equivalence relation. 
It is also straightforward to prove the following proposition.

\begin{proposition}\label{prop: DJKHS}
Let $B$ be a fuzzy partition of $X$ and $f: X \times X \to [0,1]$ a fuzzy equivalence relation. 
Then, $f$ is the fuzzy equivalence relation corresponding to $B$ iff $B$ is the fuzzy partition corresponding to $f$.
\end{proposition}

\begin{figure}
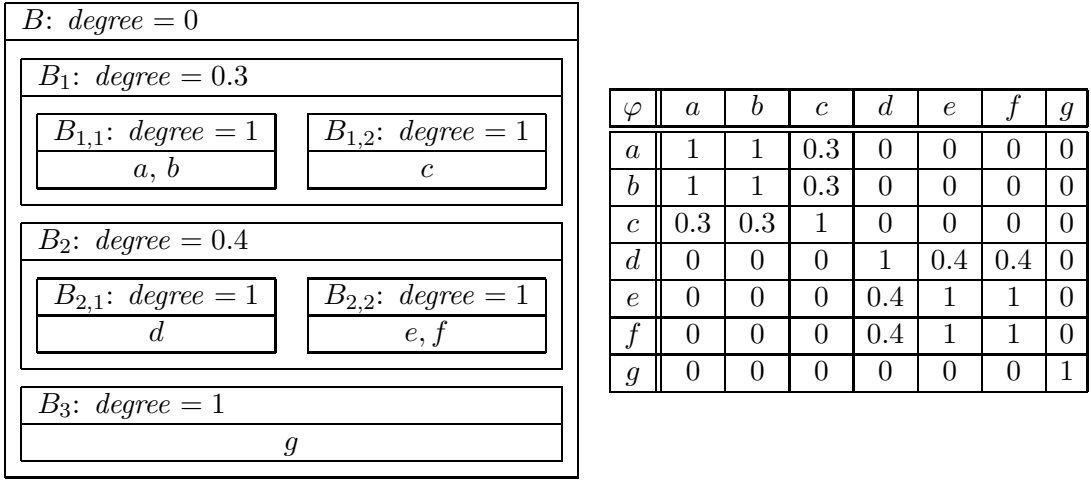

\begin{center}
\begin{tabular}{ll}
	\begin{tabular}{|l|}
		\hline
		$B$: $\degree = 0$ \\
		\hline
		\\[-1.5ex]
		\begin{tabular}{|ll|}
			\hline
			\multicolumn{2}{|l|}{$B_1$: $\degree = 0.3$} \\
			\hline
			& \\[-1.5ex]
			\begin{tabular}{|c|}
				\hline
				$B_{1,1}$: $\degree = 1$ \\
				\hline
				$a$, $b$\\
				\hline
			\end{tabular} 
			&
			\begin{tabular}{|c|}
				\hline
				$B_{1,2}$: $\degree = 1$ \\
				\hline
				$c$\\
				\hline
			\end{tabular} 
			\\[2.8ex]
			\hline
		\end{tabular} 
		\\[5.8ex]
		\begin{tabular}{|ll|}
			\hline
			\multicolumn{2}{|l|}{$B_2$: $\degree = 0.4$} \\
			\hline
			& \\[-1.5ex]
			\begin{tabular}{|c|}
				\hline
				$B_{2,1}$: $\degree = 1$ \\
				\hline
				$d$\\
				\hline
			\end{tabular} 
			&
			\begin{tabular}{|c|}
				\hline
				$B_{2,2}$: $\degree = 1$ \\
				\hline
				$e, f$\\
				\hline
			\end{tabular} 
			\\[2.8ex]
			\hline
		\end{tabular} 
		\\[5.8ex]
		\begin{tabular}{|c|}
			\hline
			$B_{3}$: $\degree = 1$ \mbox{\hspace{10.6em}} \\
			\hline
			$g$\\
			\hline
		\end{tabular} 
		\\[2.8ex]
		\hline
	\end{tabular}
	&
	\(
	\begin{array}{|c||c|c|c|c|c|c|c|}
	\hline
	\varphi & a & b & c & d & e & f & g \\
	\hline
	\hline
	a & 1 & 1 &0.3& 0 & 0 & 0 & 0 \\  
	\hline
	b & 1 & 1 &0.3& 0 & 0 & 0 & 0 \\  
	\hline
	c &0.3&0.3& 1 & 0 & 0 & 0 & 0 \\  
	\hline
	d & 0 & 0 & 0 & 1 &0.4&0.4& 0 \\  
	\hline
	e & 0 & 0 & 0 &0.4& 1 & 1 & 0 \\  
	\hline
	f & 0 & 0 & 0 &0.4& 1 & 1 & 0 \\  
	\hline
	g & 0 & 0 & 0 & 0 & 0 & 0 & 1 \\
	\hline
	\end{array}
	\)
\end{tabular}
\end{center}
\caption{An illustration for Example~\ref{example: JHDBQ}.\label{fig: HGDHW}}
\end{figure}

\begin{example}\label{example: JHDBQ}  
Let $X = \{a,b,c,d,e,f,g\}$.
Let $B$ be the fuzzy partition of $X$ depicted at the left hand side of Figure~\ref{fig: HGDHW} and specified as follows:
\begin{itemize}
\item $B$ is a fuzzy block, with $B.\degree = 0$ and $B.\subblocks() = \{B_1,B_2,B_3\}$;
\item $B_1$ is a fuzzy block, with $B_1.\degree = 0.3$ and $B_1.\subblocks() = \{B_{1,1},B_{1,2}\}$;
	\begin{itemize}
	\item $B_{1,1}$ is a crisp block, with $B_{1,1}.\elements() = \{a,b\}$;
	\item $B_{1,2}$ is a crisp block, with $B_{1,2}.\elements() = \{c\}$;
	\end{itemize}
\item $B_2$ is a fuzzy block, with $B_2.\degree = 0.4$ and $B_2.\subblocks() = \{B_{2,1},B_{2,2}\}$;
	\begin{itemize}
	\item $B_{2,1}$ is a crisp block, with $B_{2,1}.\elements() = \{d\}$;
	\item $B_{2,2}$ is a crisp block, with $B_{2,2}.\elements() = \{e,f\}$;
	\end{itemize}
\item $B_3$ is a crisp block, with $B_3.\elements() = \{g\}$.
\end{itemize}
The fuzzy block $B$ is denoted by $\{\{\{a,b\}_1,\{c\}_1\}_{0.3}, \{\{d\}_1,\{e,f\}_1\}_{0.4}, \{g\}_1\}_0$. 

Let $\varphi: X \times X \to [0,1]$ be the fuzzy relation specified at the right hand side of Figure~\ref{fig: HGDHW}. 
It is the fuzzy equivalence relation corresponding to the fuzzy partition~$B$. 

Let $X_{a,b}$, $X_c$, $X_d$, $X_{e,f}$ and $X_g$ be the fuzzy subsets of $X$ specified as follows:
\begin{itemize}
	\item $X_{a,b} = \{a\!:\!1, b\!:\!1, c\!:\!0.3\}$,
	\item $X_c = \{a\!:\!0.3, b\!:\!0.3, c\!:\!1\}$,
	\item $X_d = \{d\!:\!1, e\!:\!0.4, f\!:\!0.4\}$, 
	\item $X_{e,f} = \{d\!:\!0.4, e\!:\!1, f\!:\!1\}$,
	\item $X_g = \{g\!:\!1\}$.
\end{itemize}
They are the fuzzy equivalence classes of the fuzzy equivalence relation $\varphi$. 
Using the notion of fuzzy partition defined in \cite{OVCHINNIKOV1991107,DBLP:conf/ismvl/Schmechel95,DBLP:journals/isci/BaetsCK98,DBLP:journals/fss/CiricIB07}, the fuzzy partition corresponding to the fuzzy equivalence relation $\varphi$ is the set $\{X_{a,b}, X_c, X_d, X_{e,f}, X_g\}$. The fuzzy block $B$ considered in this example gives a more compact representation for this fuzzy partition. 
\myend
\end{example}


\begin{figure*}
\begin{procedure}[H]
\caption{ConvertFP2FB($B$)\label{proc: ConvertFP2FB}}
\Input{a fuzzy partition $B$ of a finite set $X$.}
\Output{the fuzzy equivalence relation corresponding to $B$.}
\BlankLine
declare $f: X \times X \to [0,1]$\;
$\AuxConvertFPtoFB(B,B,0,f)$\tcp*{auxiliary, defined below}
\Return $f$\;
\end{procedure}

\smallskip

\begin{procedure}[H]
\caption{AuxConvertFP2FB($B_1,B_2,d,f$)}
\uIf{$B_1 = B_2$ (as references)}{
	\uIf{$B_1$ is a crisp block}{
		\ForEach{$x \in B_1.\elements()$}{
			\ForEach{$x' \in B_2.\elements()$}{
				$f[x,x'] := 1$\;
			}
		}
	}
	\Else{
		\ForEach{$B'_1 \in B_1.\subblocks()$}{
			\ForEach{$B'_2 \in B_2.\subblocks()$}{
				$\AuxConvertFPtoFB(B'_1,B'_2, B_1.\degree,f)$\;
			}
		}
	}
}
\uElseIf{$B_1$ is a crisp block}{
	\uIf{$B_2$ is a crisp block}{
		\ForEach{$x \in B_1.\elements()$}{
			\ForEach{$x' \in B_2.\elements()$}{
				$f[x,x'] = d$\;
			}
		}
	}
	\Else{
		\ForEach{$B'_2 \in B_2.\subblocks()$}{
			$\AuxConvertFPtoFB(B_1,B'_2,d,f)$\;
		}
	}
}
\Else{
	\ForEach{$B'_1 \in B_1.\subblocks()$}{
		$\AuxConvertFPtoFB(B'_1,B_2,d,f)$\;
	}
}
\end{procedure}
\end{figure*}

The procedure $\ConvertFPtoFB(B)$ (on page \pageref{proc: ConvertFP2FB}), for a given fuzzy partition~$B$ of a finite set~$X$, returns the fuzzy equivalence relation corresponding to~$B$. It uses the subroutine $\AuxConvertFPtoFB(B_1,B_2,d,f)$, with the assumptions that either $B_1 = B_2$ or $B_1.\elements() \cap B_2.elements() = \emptyset$, and $d$ is essential only when $B_1 \neq B_2$. 

\begin{proposition}\label{prop: RHJQK}\ 
\begin{enumerate}
\item The procedure $\ConvertFPtoFB(B)$ is correct. That is, it returns the fuzzy equivalence relation corresponding the given fuzzy partition $B$.
\item Assume that the collection of elements of any crisp block (respectively, the collection of subblocks of any fuzzy block) is represented as a vector, a list or a doubly linked list. If $B$ is a fuzzy partition of a finite set $X = \{0,\ldots,n-1\}$, then the procedure $\ConvertFPtoFB(B)$ runs in time $O(n^2)$. 
\end{enumerate}
\end{proposition}

\begin{proof}
The proof for the first assertion is straightforward. For the second assertion, assume that $B$ is a fuzzy partition of a finite set $X = \{0,\ldots,n-1\}$. By induction, it can be proved that the auxiliary procedure $\AuxConvertFPtoFB(B_1,B_2,d,f)$ runs in time $O(|B_1.\elements()| \cdot |B_2.\elements()|)$. Hence, the procedure $\ConvertFPtoFB(B)$ runs in time $O(n^2)$. 
\myend
\end{proof}


We define the {\em representation tree} of a block $B$ to be the tree consisting of 
\begin{itemize}
	\item the only node $B$ if $B$ is a crisp block,
	\item the root $B$ and the subtrees that represent the subblocks of $B$ otherwise.
\end{itemize}

The {\em leaf partition} of a block $B$ is defined to be the set $\{B'.\elements() \mid$ $B'$ is a leaf of the representation tree of $B\}$. It is a crisp partition of the set $B.\elements()$. 

Given $d \in [0,1)$, the {\em $d$-cut partition} of a block $B$ is the crisp partition $\bbP$ of $B.\elements()$ defined inductively as follows:
\begin{itemize}
	\item if $B$ is a crisp block, then $\bbP = \{B.\elements()\}$;
	\item else if $B$ is a fuzzy block with $B.\degree > d$, then $\bbP = \{B.\elements()\}$;
	\item else $\bbP$ is the union of the $d$-cut partitions of the subblocks of $B$.
\end{itemize}

\begin{example}
Reconsider the fuzzy block $B$ given in Example~\ref{example: JHDBQ}. 
The leaf partition of $B$ is $\{\{a,b\},\{c\},\{d\},\{e,f\},\{g\}\}$, which is also the $0.4$-cut partition of $B$. 
The $0.3$-cut partition of $B$ is $\{\{a,b\},\{c\},\{d,e,f\},\{g\}\}$, which is also the $0.35$-cut partition of $B$. 
\myend
\end{example}

\begin{lemma}\label{lemma: JHLWB}
	Let $B$ be a fuzzy partition of a set $X$ and $f$ the fuzzy equivalence relation corresponding to $B$. Let $d \in [0,1)$ and $x,x' \in X$. Then $x$ and $x'$ belong to the same component of the $d$-cut partition of $B$ iff $f(x,x') > d$.
\end{lemma}

\begin{proof}
The elements $x$ and $x'$ belong to the same component of the $d$-cut partition of $B$ iff there exists a block $B'$ of the representation tree of $B$ such that $B'.\degree > d$ and $x,x' \in B'.\elements()$, which in turn holds iff $f(x,x') > d$. 
\myend
\end{proof}

Given blocks $B$ and $B'$ such that $B.\elements() = B'.\elements()$, we say that $B$ is {\em coarser} than $B'$ if, for every $d \in [0,1)$, the $d$-cut partition of $B$ is coarser than the $d$-cut partition of $B'$. If $B$ is coarser than $B'$, then we call $B'$ a {\em refinement} of~$B$. 
The following proposition follows immediately from Lemma~\ref{lemma: JHLWB}. 

\begin{proposition}
Let $B$ and $B'$ be fuzzy partitions of a set $X$ and let $f$ and $f'$ be the fuzzy equivalence relations corresponding to $B$ and $B'$, respectively. Then, $B$ is coarser than $B'$ iff $f \geq f'$.
\end{proposition}

It is known that, if $\bbP$ and $\bbQ$ are crisp partitions of a set $X$, $\bbP$ is coarser than $\bbQ$ and vice versa, then $\bbP$ and $\bbQ$ are equal. The following lemma is a generalization of this.

\begin{lemma}\label{lemma: HDKJA}
Let $B$ and $B'$ be fuzzy partitions of a set $X$. If $B$ is coarser than $B'$ and vice versa, then $B$ and $B'$ are equal. 
\end{lemma}

\begin{proof}
	Suppose $B$ is coarser than $B'$ and vice versa. Thus, for every $d \in [0,1)$, the $d$-cut partitions of $B$ and $B'$ are equal. We prove that $B$ and $B'$ are equal by induction on the structure of $B$ and the size of $B.\elements()$. 
	
	Consider the case when $B$ is a crisp block. If $B'$ is a fuzzy block, then, for any $d \in [B'.\degree,1)$, the $d$-cut partition of $B'$ consists of at least two components and cannot be coarser than the $d$-cut partition of $B$ (which consists of only one component), and therefore, $B'$ is not coarser than $B$. Hence, $B'$ must also be a crisp block. Since $B$ and $B'$ are crisp blocks and each of them is coarser than the other, they must be equal. 
	
	The case when $B'$ is a crisp block can be dealt with analogously. 
	
	Consider the case when both $B$ and $B'$ are fuzzy blocks. Let $d = B.\degree$ and $d' = B'.\degree$. If $d \neq d'$, then, for $d'' = \min\{d,d'\}$, the $d''$-cut partitions of $B$ and $B'$ are not equal. Hence, it must hold that $d = d'$. Since the $d$-cut partitions of $B$ and $B'$ are equal, for each $B_i \in B.\subblocks()$, there must exist $B'_i \in B'.\subblocks()$ such that $B_i.\elements() = B'_i.\elements()$. Consider an arbitrary $B_i \in B.\subblocks()$ and let $B'_i \in B'.\subblocks()$ be the block such that $B'_i.\elements() = B_i.\elements()$. 
	Let $d_i = B_i.\degree$ and $d'_i = B'_i.\degree$. If $d_i \neq d'_i$ then, for $d''_i = \min\{d_i,d'_i\}$, the $d''_i$-cut partitions of $B$ and $B'$ are not equal. Hence, it must hold that $d_i = d'_i$. 
	Since $B$ is coarser than $B'$ and vice versa, it must hold that $B_i$ is coarser than $B'_i$ and vice versa. By the induction assumption, $B_i$ and $B'_i$ are equal. 
	Thus, we can already conclude that $B$ and $B'$ are equal. 
	\myend
\end{proof}

\section{The Skeleton of the Algorithm}
\label{sec: skeleton alg}

In this section, we present an algorithm for constructing the fuzzy partition corresponding to the greatest fuzzy bisimulation of a given finite fuzzy graph. It is formulated on an abstract level with the aim to facilitate understanding the algorithm and prove its correctness. Implementation details and a complexity analysis for the algorithm are presented in the next section.

To formulate our algorithm we need an additional subclass of the abstract class {\em block}, which is named {\em simple block}. Simple blocks differ from crisp blocks in that they do not have the attribute $\degree$. They are used only during the construction of the fuzzy partition. 
The notions of representation tree of a block, leaf partition of a block and $d$-cut partition of a block are revised by treating simple blocks as crisp blocks (i.e., by replacing the occurrences of ``crisp block'' with ``crisp or simple block'').

The notion of equality between blocks is revised naturally (if $B$ or $B'$ is a simple block, then they are equal only if both of them are simple blocks and $B.\elements() = B'.\elements()$). Lemma~\ref{lemma: HDKJA} has been formulated and proved for the case when all the leaves of the representation trees of $B$ and $B'$ are crisp blocks. A variant of that lemma for the case when all the leaves of the representation trees of $B$ and $B'$ are simple blocks is given below. Its proof is obtained from the proof of Lemma~\ref{lemma: HDKJA} by replacing the occurrences of ``crisp block'' with ``simple block''.

\begin{lemma}\label{lemma: HDKJA2}
Let $B$ and $B'$ be blocks such that $B.\elements() = B'.\elements()$ and all the leaves of the representation trees of $B$ and $B'$ are simple blocks. If $B$ is coarser than $B'$ and vice versa, then $B$ and $B'$ are equal. 
\end{lemma}

In this section, let $G = \tuple{V, E, L, \SV, \SE}$ be a finite fuzzy graph. We use $\bbP$ and $\bbQ$ to denote crisp partitions of $V$, $X$ and $Y$ to denote non-empty subsets of $V$, $r$ to denote an edge label from $\SE$, and $d$ to denote a value from $[0,1)$. 
In addition, we use $B$ and $\bbB$ to denote blocks with $B.\elements() \subseteq V$ and $\bbB.\elements() = V$. 

For $x \in V$, we denote $E(x,r,Y) = \{E(x,r,y) \mid y \in Y\}$. 

We say that $X$ is {\em $d$-cut stable w.r.t.~$Y$} if the following conditions hold for all $x,x' \in X$, $p \in \SV$ and $r \in \SE$:
\begin{eqnarray*}
L(x)(p) > d & \textrm{iff} & L(x')(p) > d \\ 
\sup E(x,r,Y) > d & \textrm{iff} & \sup E(x',r,Y) > d.
\end{eqnarray*}

We say that: 
\begin{itemize}
\item $X$ is {\em $d$-cut stable w.r.t.~$\bbQ$} if it is $d$-cut stable w.r.t.~$Y$ for all $Y \in \bbQ$; 
\item $\bbP$ is {\em $d$-cut stable w.r.t.~$Y$} if $X$ is $d$-cut stable w.r.t.~$Y$ for all $X \in \bbP$;
\item $\bbP$ is {\em $d$-cut stable w.r.t.~$\bbQ$} if it is $d$-cut stable w.r.t.~$Y$ for all $Y \in \bbQ$;
\item $B$ is {\em $d$-cut stable w.r.t.~$Y$} if the $d$-cut partition of $B$ is $d$-cut stable w.r.t.~$Y$;
\item $B$ is {\em $d$-cut stable w.r.t.~$\bbQ$} if it is $d$-cut stable w.r.t.~$Y$ for all $Y \in \bbQ$;
\item $B$ is {\em $d$-cut stable} if the $d$-cut partition of $B$ is $d$-cut stable w.r.t.\ itself.
\end{itemize}

\begin{lemma}\label{lemma: FKWMN}
Let $\bbY$ be a set of subsets of $V$. If $X$ is $d$-cut stable w.r.t.\ each $Y \in \bbY$, then $X$ is also $d$-cut stable w.r.t.\ $\bigcup \bbY$.
\end{lemma}

The proofs of this lemma and the following corollary are straightforward. 

\begin{corollary}\label{cor: HGDJA}
If $X$ is $d$-cut stable w.r.t.\ $\bbP$ and $\bbQ$ is coarser than $\bbP$, then $X$ is also $d$-cut stable w.r.t.~$\bbQ$.
\end{corollary}

\begin{lemma}\label{lemma: JHFKS}
Let $G = \tuple{V, E, L, \SV, \SE}$ be a finite fuzzy graph, $\bbB$ a fuzzy partition of $V$ and $Z$ the fuzzy equivalence relation corresponding to $\bbB$. Then, $Z$ is a fuzzy bisimulation of $G$ iff $\bbB$ is $d$-cut stable for all $d \in [0,1)$.
\end{lemma}

\begin{proof}
By definition, the leaves of the representation tree of $\bbB$ are crisp blocks. 
First, suppose that $\bbB$ is $d$-cut stable for all $d \in [0,1)$. We show that $Z$ is a fuzzy bisimulation of $G$. 

Consider Condition~\eqref{eq: FB1}. Let $x,x' \in V$, $p \in \SV$ and $d = \min\{L(x)(p), L(x')(p)\}$. For a contradiction, suppose that $Z(x,x') > (L(x)(p) \fequiv L(x')(p))$. Thus, $L(x)(p) \neq L(x')(p)$ and $Z(x,x') > d$. By Lemma~\ref{lemma: JHLWB}, it follows that $x$ and $x'$ belong to the same component $X$ of the $d$-cut partition of $\bbB$. Since $\bbB$ is $d$-cut stable, $X$ is $d$-cut stable w.r.t.\ itself. Hence, $L(x)(p) > d$ iff $L(x')(p) > d$. Since $d = \min\{L(x)(p), L(x')(p)\}$, it follows that $L(x)(p) = L(x')(p)$, a contradiction.  

Consider Condition~\eqref{eq: FB2}. Let $x,x',y \in V$, $r \in \SE$ and $d = \min\{Z(x,x'),E(x,r,y)\}$. For a contradiction, suppose that, for every $y' \in V$, $d > \min\{Z(y,y'),E(x',r,y')\}$. Thus, there exists $d' < d$ such that, for every $y' \in V$, $d' > \min\{Z(y,y'),E(x',r,y')\}$. Let $X$ (respectively, $Y$) be the component belonging to the $d'$-cut partition of $\bbB$ that contains $x$ (respectively, $y$). Since $Z(x,x') > d'$, by Lemma~\ref{lemma: JHLWB}, it must hold that $x' \in X$. Since $\bbB$ is $d'$-cut stable, $X$ is $d'$-cut stable w.r.t.~$Y$. Since $E(x,r,y) > d'$, it follows that there exists $y' \in Y$ such that $E(x',r,y') > d'$. By Lemma~\ref{lemma: JHLWB}, $Z(y,y') > d'$. Therefore, $\min\{Z(y,y'),E(x',r,y')\} > d'$, a contradiction. 

Condition~\eqref{eq: FB3} can be proved analogously as done for Condition~\eqref{eq: FB2}.

Now, suppose that $Z$ is a fuzzy bisimulation of $G$ and let $d \in [0,1)$. We show that $\bbB$ is $d$-cut stable. Let $X$ be a component of the $d$-cut partition of $\bbB$ and let $x,x' \in X$. By Lemma~\ref{lemma: JHLWB}, $Z(x,x') > d$. 

Let $p \in \SV$. Since $Z$ is a fuzzy bisimulation of $G$, we have that $Z(x,x') \leq (L(x)(p) \fequiv L(x')(p))$. Hence, $L(x)(p) = L(x')(p)$ or $d < Z(x,x') \leq \min\{L(x)(p),L(x')(p)\}$. This implies that $L(x)(p) > d$ iff $L(x')(p) > d$. 

Let $r \in \SE$ and let $Y$ be a component of the $d$-cut partition of $\bbB$. Suppose that $E(x,r,y) > d$. We need to show that there exists $y' \in Y$ such that $E(x',r,y') > d$. Since $Z$ is a fuzzy bisimulation of $G$ and $d < \min\{Z(x,x'), E(x,r,y)\}$, there exists $y' \in V$ such that $d < \min\{Z(y,y'), E(x',r,y')\}$. By Lemma~\ref{lemma: JHLWB}, $y' \in Y$, which completes the proof.
\myend
\end{proof}

Lemma~\ref{lemma: JHFKS} given above provides another look on the considered problem. 
To compute the fuzzy partition corresponding to the greatest fuzzy bisimulation of~$G$, we can start from the simple block $\bbB$ with $\bbB.\elements() = V$, then for every fuzzy value $d$ used in $G$ in the increasing order, we refine $\bbB$ to make it $d$-cut stable. 

Let $\refineA(\bbB, d, \bbP)$ and $\refineB(\bbB, d, X, Y)$ be procedures that are specified as follows: 
\begin{itemize}
\item $\refineA(\bbB, d, \bbP)$ is defined only for parameters satisfying the following assumptions:
	\begin{itemize}
	\item $\bbP$ is equal to the leaf partition of $\bbB$, 
	\item all leaves of the representation tree of $\bbB$ are simple blocks, 
	\item each inner node $B'$ of the representation tree of $\bbB$ has $B'.\degree < d$.
	\end{itemize} 
The procedure $\refineA(\bbB, d, \bbP)$ changes $\bbB$ to the coarsest refinement of $\bbB$ that is $d$-cut stable w.r.t.~$\bbP$ and whose representation tree does not use crisp blocks. 

\item $\refineB(\bbB, d, X, Y)$ is defined only for parameters satisfying the following assumptions:
\begin{itemize}
	\item $X \subset Y$, 
	\item all leaves of the representation tree of $\bbB$ are simple blocks, 
	\item each inner node $B'$ of the representation tree of $\bbB$ has $B'.\degree \leq d$.
\end{itemize} 
The procedure $\refineB(\bbB, d, X, Y)$ changes $\bbB$ to the coarsest refinement of $\bbB$ that is $d$-cut stable w.r.t.\ both $X$ and $Y \setminus X$ and whose representation tree does not use crisp blocks. 
\end{itemize}

How to implement the procedures $\refineA(\bbB, d, \bbP)$ and $\refineB(\bbB, d, X, Y)$ efficiently is postponed to the next section. We present below only a straightforward implementation of them in order to give an insight on these procedures.

The procedure $\refineA(\bbB, d, \bbP)$ can be implemented as follows:
\begin{itemize}
\item[] for each leaf $B$ of the representation tree of $\bbB$ such that $B.\elements()$ is not $d$-cut stable w.r.t.~$\bbP$:
	\begin{itemize}
	\item let $\{X_1,\ldots,X_k\}$ be the coarsest refinement of $\{B.\elements()\}$ that is $d$-cut stable w.r.t.~$\bbP$;
	\item modify $\bbB$ by replacing the leaf $B$ in its representation tree with a new node that is the fuzzy block $B'$ specified as follows: $B'.\degree = d$ and $B'.\subblocks()$ consists of the simple blocks $B'_1,\ldots,B'_k$ with $B'_i.\elements() = X_i$, for $1 \leq i \leq k$.
	\end{itemize}
\end{itemize}

\begin{proposition}\label{prop: JHDKW}
The above implementation of $\refineA(\bbB, d, \bbP)$ satisfies the specification. That is, if $\bbB_0$ is $\bbB$ before executing $\refineA(\bbB, d, \bbP)$, then after the execution, $\bbB$ is equal to the coarsest refinement of $\bbB_0$ that is $d$-cut stable w.r.t.~$\bbP$ and whose representation tree does not use crisp blocks. 
\end{proposition}

\begin{proof}
Let $\bbB_0$ be $\bbB$ before executing $\refineA(\bbB, d, \bbP)$ and let $\bbB_1$ be the coarsest refinement of $\bbB_0$ that is $d$-cut stable w.r.t.~$\bbP$ and whose representation tree does not use crisp blocks.  

We first prove that it is an invariant of the loop of the implementation that $\bbB$ is coarser than $\bbB_1$. Consider an iteration of the loop and let $B$ and $X_1,\ldots,X_k$ be the objects used in that iteration. Observe that each inner node $B''$ of the representation tree of $\bbB$ has $B''.\degree \leq d$. By the induction assumption about the invariant, before executing that iteration, the $d$-cut partition of $\bbB$ is coarser than the $d$-cut partition of $\bbB_1$, and hence, a subset $\bbX$ of the $d$-cut partition of $\bbB_1$ is a crisp partition of the set $B.\elements()$. Since $\bbB_1$ is $d$-cut stable w.r.t.~$\bbP$, $\bbX$ is $d$-cut stable w.r.t.~$\bbP$, and therefore $\bbX$ is a refinement of $\{X_1,\ldots,X_k\}$. By executing the mentioned iteration of the loop, $\bbB$ is modified so that the $d$-cut partition of $\bbB$ after the execution differs from the one before the execution in that $B.\elements()$ is replaced by $X_1,\ldots,X_k$. Therefore, the $d$-cut partition of $\bbB$ after executing the iteration is still coarser than the $d$-cut partition of $\bbB_1$. 
For $d' \in (d,1)$, after executing the iteration, the $d'$-cut partition of $\bbB$ is the same as the $d$-cut partition of $\bbB$ and is therefore coarser than the $d'$-cut partition of $\bbB_1$. 
For $d' \in [0,d)$, the $d'$-cut partition of $\bbB$ is not affected by the iteration and is therefore coarser than the $d'$-cut partition of $\bbB_1$. 
We conclude that, after executing the mentioned iteration of the loop, $\bbB$ is still coarser than $\bbB_1$. 

Clearly, the loop of the implementation terminates (because all the sets $X_1,\ldots,X_k$ are $d$-cut stable w.r.t.~$\bbP$), and at the end, $\bbB$ is $d$-cut stable w.r.t.~$\bbP$. Let $\bbB_2$ denote $\bbB$ after executing the procedure. Thus, $\bbB_2$ is a refinement of $\bbB_0$, $\bbB_2$ is $d$-cut stable w.r.t.~$\bbP$, and clearly, all leaves of the representation tree of $\bbB_2$ are simple blocks. Therefore, $\bbB_2$ is a refinement of $\bbB_1$. On the other hand, by the above mentioned invariant, $\bbB_2$ is coarser than $\bbB_1$. By Lemma~\ref{lemma: HDKJA2}, it follows that $\bbB_2$ is equal to $\bbB_1$, which completes the proof.
\myend
\end{proof}

The procedure $\refineB(\bbB, d, X, Y)$ can be implemented as follows:
\begin{itemize}\label{place HDJHS}
	\item[] for each leaf $B$ of the representation tree of $\bbB$ such that $B.\elements()$ is not $d$-cut stable w.r.t.\ both $X$ and $Y \setminus X$:
	\begin{itemize}
		\item let $\{X_1,\ldots,X_k\}$ be the coarsest refinement of $\{B.\elements()\}$ that is $d$-cut stable w.r.t.\ both $X$ and $Y \setminus X$;
		\item if either $B$ is the root of $\bbB$ or the parent node $B_p$ of the leaf $B$ in the representation tree of $\bbB$ has $B_p.\degree < d$ then:
			\begin{itemize}
			\item modify $\bbB$ by replacing the leaf $B$ in its representation tree with a new node that is the fuzzy block $B'$ specified as follows: $B'.\degree = d$ and $B'.\subblocks()$ consists of the simple blocks $B'_1,\ldots,B'_k$ with $B'_i.\elements() = X_i$, for $1 \leq i \leq k$.
			\end{itemize}
		\item else:
			\begin{itemize}
			\item let $B_p$ be the parent node of the leaf $B$ in the representation tree of $\bbB$; (by the assumptions about the parameters, we have that $B_p.\degree = d$)
			\item modify $\bbB$ by replacing $B$ in the collection of subblocks of $B_p$ with the simple blocks $B'_1,\ldots,B'_k$ with $B'_i.\elements() = X_i$, for $1 \leq i \leq k$.
			\end{itemize}
	\end{itemize}
\end{itemize}

\begin{proposition}\label{prop: JHDKW2}
The above implementation of $\refineB(\bbB, d, X, Y)$ satisfies the specification. That is, if $\bbB_0$ is $\bbB$ before executing $\refineB(\bbB, d, X, Y)$, then after the execution, $\bbB$ is equal to the coarsest refinement of $\bbB_0$ that is $d$-cut stable w.r.t.\ both $X$ and $Y \setminus X$ and whose representation tree does not use crisp blocks. 
\end{proposition}

This proposition can be proved analogously as done for Proposition~\ref{prop: JHDKW}.

\begin{corollary}\label{cor: HGFJS}
Independently from the implementation of $\refineA$ and $\refineB$ (which is required to be correct), after executing $\refineA(\bbB, d, \bbP)$ or $\refineB(\bbB, d, X, Y)$, each inner node of the representation tree of $\bbB$ is a fuzzy block $B$ with $B.\degree \leq d$.
\end{corollary}

This corollary follows immediately from Propositions~\ref{prop: JHDKW} and~\ref{prop: JHDKW2} and Lemma~\ref{lemma: HDKJA2}. 


We provide Algorithm~\ref{algCompFP} (on page~\pageref{algCompFP}) for computing the fuzzy partition corresponding to the greatest fuzzy bisimulation of~$G$. Let $d_0 = 0$ and let $d_1 < d_2 < \ldots < d_k$ be all the non-zero fuzzy values used for $G$. The algorithm starts from initializing $\bbB$ to the simple block with $\bbB.\elements() = V$. Then, for every $i$ from 0 to $k-1$, it refines $\bbB$ to make it $d_i$-cut stable. The loop uses a crisp partition $\bbP$ of $V$, which is set up at the beginning and maintained so that $\bbP$ is equal to the leaf partition of $\bbB$ before and after each iteration. In each iteration of the loop, $\bbB$ is refined by two phases. First, a call of $\refineA(\bbB,d_i,\bbP)$ modifies $\bbB$ to make it $d_i$-cut stable w.r.t.~$\bbP$. Then, while $\bbP$ differs from the leaf partition of $\bbB$, the algorithm chooses a set $X$ from the leaf partition of $\bbB$ and a set $Y$ from $\bbP$ such that $X \subset Y$, and then call $\refineB(\bbB, d_i, X, Y)$ to make $\bbB$ $d_i$-cut stable w.r.t.\ both $X$ and $Y \setminus X$ before replacing $Y$ in $\bbP$ with $X$ and $Y \setminus X$. These two phases together make $\bbB$ $d_i$-cut stable. 
Following the idea from Hopcroft's algorithm~\cite{Hopcroft71} and Paige and Tarjan's algorithm~\cite{PaigeT87}, the selected set $X$ is one that is at least twice smaller than the mentioned set $Y$. This criterion is indifferent for the correctness of the algorithm, but essential for reducing the complexity order of the algorithm. 
At the end, the algorithm modifies $\bbB$ by converting each leaf of its representation tree (from a simple block) to a crisp block and then returns~$\bbB$.


\begin{algorithm}[t]
\caption{\CompFP\label{algCompFP}}
\Input{a finite fuzzy graph $G = \tuple{V, E, L, \SV, \SE}$.}
\Output{the fuzzy partition corresponding to the greatest fuzzy bisimulation of~$G$.}
	
\smallskip

let $d_0 = 0$ and let $d_1 < d_2 < \ldots < d_k$ be all the non-zero fuzzy values used for $G$\; 
set $\bbP := \{V\}$ and set $\bbB$ to the simple block such that $\bbB.\elements() = V$\;

\ForEach{$i$ from 0 to $k-1$\label{step: HGFAJ 4}}{
	$\refineA(\bbB, d_i, \bbP)$\label{step: HGFAJ 4b}\tcp*{to make $\bbB$ $d_i$-cut stable w.r.t.~$\bbP$}
	\While{$\bbP$ differs from the leaf partition of $\bbB$}{
		choose $X$ from the leaf partition of $\bbB$ and $Y$ from $\bbP$ such that $X \subset Y$ and \mbox{$|X| \leq |Y|/2$}\;
		$\refineB(\bbB, d_i, X, Y)$\label{step: HGFAJ 7}\tcp*{to make $\bbB$ $d_i$-cut stable w.r.t.\ both $X$ and $Y \setminus X$}
		refine $\bbP$ by replacing $Y$ with $X$ and $Y \setminus X$\label{step: HGFAJ 8}\;
	}	
}

modify $\bbB$ by converting each leaf of its representation tree, which is a simple block, to a~crisp block with the same set of elements\;
\Return $\bbB$.	
\end{algorithm}


\begin{example}
	Let $G = \tuple{V, E, L, \SV, \SE}$ be the fuzzy graph depicted and specified as follows:
	
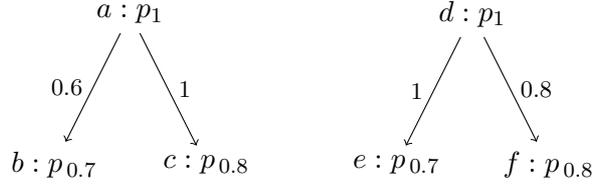
\begin{center}
		\begin{tikzpicture}
		\node (x0) {};
		\node (x) [node distance=1.5cm, right of=x0] {};
		\node (I) [node distance=0.0cm, below of=x] {};
		\node (I1) [node distance=4.5cm, right of=I] {};
		\node (u) [node distance=0.0cm, below of=I] {$a:p_1$};
		\node (ub) [node distance=2.0cm, below of=u] {};
		\node (v) [node distance=1.0cm, left of=ub] {$b:p_{\,0.7}$};
		\node (w) [node distance=1.0cm, right of=ub] {$c:p_{\,0.8}$};
		\node (up) [node distance=0.0cm, below of=I1] {$d:p_1$};
		\node (ubp) [node distance=2.0cm, below of=up] {};
		\node (vp) [node distance=1.0cm, left of=ubp] {$e:p_{\,0.7}$};
		\node (wp) [node distance=1.0cm, right of=ubp] {$f:p_{\,0.8}$};
		\draw[->] (u) to node [left]{\footnotesize{0.6}} (v);
		\draw[->] (u) to node [right]{\footnotesize{1}} (w);
		\draw[->] (up) to node [left]{\footnotesize{1}} (vp);
		\draw[->] (up) to node [right]{\footnotesize{0.8}} (wp);
		\end{tikzpicture}
	\end{center}	
	\begin{itemize}
		\item $V = \{a,b,c,d,e,f\}$, $\SV = \{p\}$, $\SE = \{r\}$, 
		\item $E = \{\tuple{a,r,b}\!:\!0.6, \tuple{a,r,c}\!:\!1, \tuple{d,r,e}\!:\!1, \tuple{d,r,f}\!:\!0.8\}$, 
		\item $L(a)(p) = L(d)(p) = 1$, $L(b)(p) = L(e)(p) = 0.7$, $L(c)(p) = L(f)(p) = 0.8$.
	\end{itemize}
	It is the disjoint union of the fuzzy graphs given in Example~\ref{example: JHFJH}. 
	Let's apply Algorithm~\ref{algCompFP} to this fuzzy graph. The effects of the steps are as follows. 
	\begin{itemize}
		\item We have $k = 4$ and $(d_0,\ldots,d_4) = (0, 0.6, 0.7, 0.8, 1)$.
		\item At the beginning, $\bbP = \{V\}$ and $\bbB$ is the simple block with $\bbB.\elements() = V$.
		
		\item Executing $\refineA(\bbB, 0, \bbP)$, $\bbB$ is changed to the fuzzy block with $\bbB.\degree = 0$ and $\bbB.\subblocks() = \{B_1,B_2\}$, where $B_1$ and $B_2$ are the simple blocks with $B_1.\elements() = \{a,d\}$ and $B_2.\elements() = \{b,c,e,f\}$.
		
		\item Executing $\refineB(\bbB, 0, X, Y)$ with $X = \{a,d\}$ and $Y = V$, $\bbB$ remains the same. After that, $\bbP$ is changed to $\{\{a,d\}, \{b,c,e,f\}\}$. 
		
		\item Executing $\refineA(\bbB, 0.6, \bbP)$, $\bbB$ remains the same.
		
		\item Executing $\refineA(\bbB, 0.7, \bbP)$, the subblock $B_2$ of $\bbB$ is changed to the fuzzy block with $B_2.\degree = 0.7$ and $B_2.\subblocks() = \{B_{2,1}, B_{2,2}\}$, where $B_{2,1}$ and $B_{2,2}$ are the simple blocks with $B_{2,1}.\elements() = \{b,e\}$ and $B_{2,2}.\elements() = \{c,f\}$. 
		
		\item Executing $\refineB(\bbB, 0.7, X, Y)$ with $X = \{b,e\}$ and $Y = \{b,c,e,f\}$, the subblock $B_1$ of $\bbB$ is changed to the fuzzy block with $B_1.\degree = 0.7$ and $B_1.\subblocks() = \{B_{1,1}, B_{1,2}\}$, where $B_{1,1}$ and $B_{1,2}$ are the simple blocks with $B_{1,1}.\elements() = \{a\}$ and $B_{1,2}.\elements() = \{d\}$. After that, $\bbP$ is changed to $\{\{a,d\}, \{b,e\},\{c,f\}\}$. 
		
		\item Executing $\refineB(\bbB, 0.7, X, Y)$ with $X = \{a\}$ and $Y = \{a,d\}$, $\bbB$ remains the same. After that, $\bbP$ is changed to $\{\{a\}, \{d\}, \{b,e\}, \{c,f\}\}$. 
		
		\item Executing $\refineA(\bbB, 0.8, \bbP)$, $\bbB$ remains the same.
		
		\item At the end, we obtain $\bbB = \{\{\{a\}_1, \{d\}_1\}_{0.7}, \{\{b,e\}_1, \{c,f\}_1\}_{0.7}\}_0$.
		\myend
	\end{itemize}
\end{example}


\begin{lemma}\label{lemma: HDJSJ}
Let $\bbB_0$ be the fuzzy partition corresponding to the greatest fuzzy bisimulation of~$G$. 
Here are invariants that hold before and after executing each enumerated statement of the ``foreach'' loop of Algorithm~\ref{algCompFP}:
\begin{enumerate}
\item $\bbP$ is coarser than the leaf partition of $\bbB$.
\item The leaf partition of $\bbB$ is equal to the $d_i$-cut partition of $\bbB$.
\item $\bbB$ is coarser than $\bbB_0$. 
\end{enumerate}
\end{lemma} 

\begin{proof}
Observe that, during the run of Algorithm~\ref{algCompFP}, $\bbB$ is refined by splitting leaves of the representation tree and converting some of them to inner nodes. Clearly, the first invariant holds before the ``foreach'' loop. During the loop, $\bbP$ is modified by replacing a block $Y \in \bbP$ with blocks $X$ and $Y \setminus X$, where $X$ is a component of the leaf partition of $\bbB$ and $X \subset Y$. Therefore, $\bbP$ is always coarser than the leaf partition of~$\bbB$. 

By Corollary~\ref{cor: HGFJS}, during the ``foreach'' loop, each inner node of the representation tree of $\bbB$ is a fuzzy block $B$ such that $B.\degree \leq d_i$. Therefore, the second invariant holds. 

Consider the third invariant. Let $\bbB'$ (respectively, $\bbB''$) be $\bbB$ before (respectively, after) executing one of the involved calls of $\refineA$ or $\refineB$. As the induction assumption, $\bbB'$ is coarser than $\bbB_0$, and for the induction hypothesis, we need to prove that $\bbB''$ is also coarser than~$\bbB_0$.  
\begin{itemize}
\item Consider the call of $\refineA(\bbB, d_i, \bbP)$ at the statement~\ref{step: HGFAJ 4b}. By Lemma~\ref{lemma: JHFKS}, $\bbB_0$ is $d_i$-cut stable. That is, the $d_i$-cut partition of $\bbB_0$ is $d_i$-cut stable w.r.t.\ itself. Since $\bbB'$ is coarser than $\bbB_0$, by Corollary~\ref{cor: HGDJA}, the $d_i$-cut partition of $\bbB_0$ is $d_i$-cut stable w.r.t.\ the $d_i$-cut partition of $\bbB'$. By the first two invariants and by Corollary~\ref{cor: HGDJA} once again, it follows that the $d_i$-cut partition of $\bbB_0$ is $d_i$-cut stable w.r.t.~$\bbP$. Since $\bbB'$ is coarser than $\bbB_0$, by the specification of $\refineA$, $\bbB''$ is also coarser than $\bbB_0$. 

\item Consider the call of $\refineB(\bbB, d_i, X, Y)$ at the statement~\ref{step: HGFAJ 7}. 
By Lemma~\ref{lemma: JHFKS}, $\bbB_0$ is $d_i$-cut stable. That is, the $d_i$-cut partition of $\bbB_0$ is $d_i$-cut stable w.r.t.\ itself. Since $\bbB'$ is coarser than $\bbB_0$, by Corollary~\ref{cor: HGDJA}, the $d_i$-cut partition of $\bbB_0$ is $d_i$-cut stable w.r.t.\ the $d_i$-cut partition of $\bbB'$. By the first two invariants and by Lemma~\ref{lemma: FKWMN}, it follows that the $d_i$-cut partition of $\bbB_0$ is $d_i$-cut stable w.r.t.\ both $X$ and $Y \setminus X$. Since $\bbB'$ is coarser than $\bbB_0$, by the specification of $\refineB$, $\bbB''$ is also coarser than $\bbB_0$.
\end{itemize}
We have shown that the third invariant holds. This completes the proof. 
\myend
\end{proof}

\begin{lemma}\label{lemma: DJSJA}
The following assertions hold:
\begin{enumerate}
\item It is an invariant of the ``while'' loop of Algorithm~\ref{algCompFP} that $\bbB$ is $d_i$-cut stable w.r.t.~$\bbP$. 
\item After executing the ``while'' loop in Algorithm~\ref{algCompFP}, $\bbB$ is $d_i$-cut stable. 
\item It is an invariant of the ``foreach'' loop of Algorithm~\ref{algCompFP} that $\bbB$ is $d_j$-cut stable for all $0 \leq j < i$. 
\end{enumerate}
\end{lemma}

\begin{proof}
Consider the first assertion. The invariant holds at each moment immediately before executing the ``while'' loop, because such a moment occurs after executing $\refineA(\bbB, d_i, \bbP)$. The invariant holds after each iteration of the ``while'' loop due to the specification of $\refineB$ and the modification of $\bbP$.

Consider the second assertion. After executing the ``while'' loop in Algorithm~\ref{algCompFP}, we have that: $\bbP$ is equal to the leaf partition of $\bbB$, and by the second assertion of Lemma~\ref{lemma: HDJSJ}, $\bbP$ is equal to the $d_i$-cut partition of $\bbB$; hence, by the first assertion of the current lemma, $\bbB$ is $d_i$-cut stable.

Consider the third assertion. By Propositions~\ref{prop: JHDKW} and~\ref{prop: JHDKW2} and Lemma~\ref{lemma: HDKJA2}, the procedures $\refineA(\bbB, d, \bbP)$ and $\refineB(\bbB, d, X, Y)$ refine $\bbB$ by splitting leaves of the representation tree of $\bbB$, which are simple blocks, and convert some of them to fuzzy blocks with the $\degree$ attribute set to $d$. As the ``foreach'' loop of Algorithm~\ref{algCompFP} considers $d_i$ in the increasing order, once $\bbB$ becomes $d_i$-cut stable (as mentioned in the second assertion), it remains $d_i$-cut stable until the end of the algorithm. Therefore, it is an invariant of the ``foreach'' loop that $\bbB$ is $d_j$-cut stable for all $0 \leq j < i$. 
\myend
\end{proof}

\begin{theorem}\label{theorem: correctness 1}
Algorithm~\ref{algCompFP} is correct. That is, it is a correct algorithm for computing the fuzzy partition corresponding to the greatest fuzzy bisimulation of a given finite fuzzy graph.
\end{theorem}

\begin{proof}
Let $\bbB$ be the fuzzy partition returned by Algorithm~\ref{algCompFP}. Recall that $d_0 = 0$ and $d_1 < d_2 < \ldots < d_k$ are all the non-zero fuzzy values used for $G$. By Propositions~\ref{prop: JHDKW} and~\ref{prop: JHDKW2} and Lemma~\ref{lemma: HDKJA2}, for each fuzzy block $B$ occurring in the representation tree of $\bbB$, $B.\degree \in \{d_0,\ldots,d_{k-1}\}$. 
By Lemma~\ref{lemma: DJSJA}, $\bbB$ is $d_i$-cut stable for all $0 \leq i < k$. 
For $0 \leq i < k$ and $d \in [d_i,d_{i+1})$, $\bbB$ is $d$-cut stable because it is $d_i$-cut stable. 
By definition, $\bbB$ is $d$-cut stable for all $d \in [d_k,1)$.
Therefore, $\bbB$ is $d$-cut stable for all $d \in [0,1)$. 
By Lemma~\ref{lemma: JHFKS}, it follows that the fuzzy equivalence relation corresponding to $\bbB$ is a fuzzy bisimulation of $G$. By Proposition~\ref{prop: DJKHS} and the third assertion of Lemma~\ref{lemma: HDJSJ}, it follows that $\bbB$ is the fuzzy partition corresponding to the greatest fuzzy bisimulation of~$G$.
\myend
\end{proof}


\newcommand{\FBlock}{\mathit{FBlock}}
\newcommand{\SCBlock}{\mathit{SCBlock}}
\newcommand{\pComponent}{\mathit{pComponent}}
\newcommand{\PComponent}{\mathit{PComponent}}
\newcommand{\PComponentList}{\mathit{PComponentList}}
\newcommand{\PPartition}{\mathit{PPartition}}
\newcommand{\vertexLabel}{\mathit{vertexLabel}}
\newcommand{\vLabel}{\mathit{label}}
\newcommand{\vLabelB}{\mathit{label}_2}
\newcommand{\vLabelC}{\mathit{label}_3}
\newcommand{\vLabelD}{\mathit{label}_4}
\newcommand{\vLabelE}{\mathit{label}_5}
\newcommand{\crisp}{\mathit{crisp}}
\newcommand{\Set}{\mathit{Set}}
\newcommand{\SCBlockList}{\mathit{SCBlockList}}
\newcommand{\scBlocks}{\mathit{scBlocks}}
\newcommand{\compoundComponents}{\mathit{compoundComponents}}
\newcommand{\simpleComponents}{\mathit{simpleComponents}}

\newcommand{\aboutVertex}{\mathit{aboutVertex}}
\newcommand{\vertex}{\mathit{vertex}}
\newcommand{\edge}{\mathit{edge}}

\section{Implementation Details and Complexity Analysis}
\label{sec: impl}

In this section, we show how to implement Algorithm~\ref{algCompFP} so that its complexity is of order $O((m\log{l} + n)\log{n})$, where $n$, $m$ and $l$ are the number of vertices, the number of non-zero edges and the number of different fuzzy degrees
of edges of the input graph $G$, respectively. 
To facilitate a full understanding of the implementation and its complexity analysis, we use the object-oriented approach and describe the data structures in detail.

\subsection{Data Structures}
\label{sec: DS}

We describe how to get an efficient implementation of Algorithm~\ref{algCompFP} by using a number of classes. In the description given below, we refer to the input graph $G = \tuple{V, E, L, \SV, \SE}$ and the variables $\bbB$ and $\bbP$ used in the algorithm. The classes are listed below:

\begin{itemize}
\item $\Vertex$: the type for the vertices of $G$; 
\item $\Edge$: the type for the edges of $G$; 
\item $\Block$: the type for blocks, which are nodes of the representation tree of $\bbB$;
\item $\FBlock$: the type for fuzzy blocks, which is a subtype of $\Block$;
\item $\SCBlock$: the type for simple or crisp blocks, which is a subtype of $\Block$;
\item $\PComponent$: the type for the components of the partition $\bbP$;
\item $\PPartition$: the type for the partition $\bbP$;
\item $\PComponentEdge$: the type for objects specifying information about edges connecting a vertex to (vertices belonging to) a component of the partition $\bbP$;
\item $\LabelDegree$: the type for objects representing the degree of a vertex label or an edge together with information about that vertex label or edge;
\item $\VertexList$, $\BlockList$, $\SCBlockList$ and $\PComponentList$: the types for doubly linked lists of elements of the type $\Vertex$, $\Block$, $\SCBlock$ or $\PComponent$, respectively; 
\item $\EdgeList$: the type for lists of elements of the type $\Edge$. 
\end{itemize}

We call objects of the type $\FBlock$, $\SCBlock$, $\PComponent$ or $\PComponentEdge$ fuzzy blocks, simple or crisp blocks, $\bbP$-components and $\bbP$-component-edges, respectively. We give below details for nontrivial classes in the above list. As in the Java language, attributes of objects are primitive values or references. 

\bigskip

\noindent
{\bf Vertex}. This class has the following instance attributes. 
\begin{itemize}
\item $\vertexID$ is the ID of the vertex (a natural number or a string).
\item $\vLabel: \SV \to (0,1]$ is a map such that, for $p \in \SV$, $\vLabel[p]$ means the degree in which $p$ is a member of the label of the vertex. 
\item \mbox{$\vertexBlock: \SCBlock$} is a reference to a leaf of the representation tree of $\bbB$. 
\item \mbox{$\nekst: \Vertex$ and $\prev: \Vertex$} are the next vertex and the previous vertex in the doubly linked list that contains the current vertex.
\item \mbox{$\comingEdges: \EdgeList$} is the list of edges coming to the vertex. 
\item \mbox{$\processed: \bool$} is an auxiliary flag for internal processing. 
\item \mbox{$\vLabelB: \Set(\SV)$} is a subset of $\SV$, used as an auxiliary attribute for implementing the procedure $\refineA$.
\item \mbox{$\vLabelC: \Set(\SE \times \PComponent)$} is a set of pairs of type $\SE \times \PComponent$, used as an auxiliary attribute for implementing the procedure $\refineA$.
\item \mbox{$\vLabelD: \Set(\SE)$} and \mbox{$\vLabelE: \Set(\SE)$} are subsets of $\SE$, used as auxiliary attributes for implementing the procedure $\refineB$.
\end{itemize}

The constructor $\Vertex(id')$ sets $\vertexID$ to $id'$, $\vLabel$ to an empty map, $\vertexBlock$, $\nekst$ and $\prev$ to $\Null$, $\comingEdges$ to a newly created empty list, $\processed$ to $\false$, $\vLabelB$, $\vLabelC$, $\vLabelD$ and $\vLabelE$ to newly created empty sets. The class has a method $\addLabel(p,d)$, which sets $\vLabel[p]$ to~$d$. It also has a static method $\getVertex(id)$ that returns the vertex with the given ID. It uses a class attribute to store the collection of the vertices that have been created so far. 

\bigskip

\noindent
{\bf Edge}. This class has the following instance attributes. 
\begin{itemize}
\item $\edgeLabel: \SE$
\item $\edgeOrigin: \Vertex$
\item $\edgeDest: \Vertex$
\item \mbox{$\degree: (0,1]$} is the value of $E(x,r,y)$, where $x$, $r$ and $y$ are the $\edgeOrigin$, $\edgeLabel$ and $\edgeDest$ of the edge, respectively. 
\item \mbox{$\pcEdge: \PComponentEdge$} specifies information about the set of edges labeled by $r$ from $x$ to the vertices of $Y$, where $r$ and $x$ are the $\edgeLabel$ and $\edgeOrigin$ of the current edge, respectively, and $Y$ is the $\bbP$-component that contains the $\edgeDest$ of the current edge. 
\end{itemize}

The constructor $\Edge(r,x,y,d,pce)$ sets the above listed attributes to the parameters, respectively, and then adds the current edge to the list $\edgeDest.\comingEdges$. 

\bigskip

\noindent
{\bf Block}. This is an abstract class. 
It has the following instance attributes. 
\begin{itemize}
\item \mbox{$\parent: Block$} is a reference to the parent node in the representation tree of $\bbB$. 
\item \mbox{$\nekst: \Block$} and \mbox{$\prev: \Block$} are the next block and the previous block in the doubly linked list of the type $\BlockList$ that contains the current block.
\end{itemize}

\bigskip

\noindent
{\bf FBlock}. This is a subclass of the class $\Block$. Recall that it is the type for fuzzy blocks. It has the following instance attributes. 
\begin{itemize}
\item \mbox{$\degree: [0,1)$} is the degree of the fuzzy block.
\item \mbox{$\subblocks: \BlockList$} is a doubly linked list of the subblocks of the fuzzy block. 
\end{itemize}

The constructor $\FBlock(\parent',\degree')$ sets $\parent$ to $\parent'$, $\degree$ to $\degree'$, $\subblocks$ to a newly created empty list, and $\nekst$ and $\prev$ to $\Null$. If $\parent' \neq \Null$, then it also adds the current block ($this$) to $\parent'.\subblocks$. 

\bigskip

\noindent
{\bf SCBlock}. This is a subclass of the class $\Block$. Recall that it is the type for simple or crisp blocks. It has the following instance attributes. 
\begin{itemize}
\item \mbox{$\crisp: \bool$} specifies whether the block is crisp or simple.
\item \mbox{$\elements: \VertexList$} is a doubly linked list of the elements of the block. 
\item \mbox{$\pComponent: \PComponent$} is a reference to the $\bbP$-component that is a superset of the current block, when treating them as sets of vertices. 
\item \mbox{$\nekst_2: \SCBlock$} and \mbox{$\prev_2: \SCBlock$} are the next block and the previous block in the doubly linked list of the type $\SCBlockList$ that contains the current block.
\item \mbox{$\departingBlocks_1: \Set(\SV) \times \Set(\SE \times \PComponent) \to \VertexList$} is an auxiliary attribute used for implementing the procedure $\refineA$. 
\item \mbox{$\departingBlocks_2: \Set(\SE) \times \Set(\SE) \to \VertexList$} is an auxiliary attribute used for implementing the procedure $\refineB$.
\end{itemize}

\begin{algorithm}[t]
	\SetAlgoVlined
	\SetKwProg{constructor}{Constructor}{:}{}
	
	\constructor{$\SCBlock(\parent'$, $\elements'$, $\pComponent')$\label{contrSCBlock}}{
		$\crisp := \false$\; 
		$\parent := \parent'$, $\elements := \elements'$, $\pComponent := \pComponent'$\; 
		set $\nekst_2$ and $\prev_2$ to $\Null$\;
		set $\departingBlocks_1$ and $\departingBlocks_2$ to newly created empty maps\;
		\uIf{$\parent = \Null$}{$\nekst := \Null$, $\prev := \Null$\;}
		\Else{add $this$ to $\parent.\subblocks$\tcp*{setting $\nekst$ and $\prev$ appropriately}} 
		\ForEach{$x \in \elements$}{$x.\vertexBlock := this$\;}
		\If{$\pComponent \neq \Null$}{$\pComponent.\addBlock(this)$\;}
	} 
\end{algorithm}

The constructor $\SCBlock(\parent'$, $\elements'$, $\pComponent')$ is presented on page~\pageref{contrSCBlock}.

\bigskip

\noindent
{\bf PComponent}. Recall that this class is the type for the components of the partition $\bbP$. It has the following instance attributes. 
\begin{itemize}
	\item \mbox{$\scBlocks: \SCBlockList$} is a list of blocks that are leaves of the representation tree of $\bbB$, which form a partition of the component (when treating the blocks and the component as sets of vertices).  
	\item \mbox{$\pPartition: \PPartition$} is a reference to $\bbP$. 
	\item \mbox{$\nekst: \PComponent$ and $\prev: \PComponent$} are the next component and the previous component in the doubly linked list that contains the current component.
\end{itemize}

The constructor $\PComponent(\pPartition')$ sets $\pPartition$ to $\pPartition'$, $\scBlocks$ to a newly created empty list, $\nekst$ and $\prev$ to $\Null$, and adds the current component to $\pPartition$ by calling $\pPartition.\addPComponent(this)$. The class $\PComponent$ has the following methods.
\begin{itemize}
\item $\Size()$ is the method that returns the number of blocks in the list $\scBlocks$. 

\item $\compound()$ is the method that returns the truth of $\Size() > 1$. A $\bbP$-component, as an object of type $\PComponent$, is said to be {\em compound} (respectively, {\em simple}) if its method $\compound()$ returns $\true$ (respectively, $\false$).
	
\item $\smallerBlock()$ is a method that can be called only when the current component is compound. It compares the first two blocks of the current component and returns the smaller one (or any one when their sizes are equal).

\item $\addBlock(b)$ is the method that adds the simple block $b$ to the list $\scBlocks$. If the addition causes that $\Size() = 2$, then the method also moves the current component from $\pPartition.\simpleComponents$ to $\pPartition.\compoundComponents$. 

\item $\removeBlock(b)$ is the method that removes the simple block $b$ from the list $\scBlocks$. If the removal causes that $\Size() = 1$, then the method also moves the current component from $\pPartition.\compoundComponents$ to $\pPartition.\simpleComponents$. 

\item $\createPComponent(\pPartition',b)$ is a static method that executes the statements $pc := \New \PComponent(\pPartition')$, $pc.\addBlock(b)$, and $b.\pComponent := pc$.
\end{itemize}

\bigskip

\noindent
{\bf PPartition}. Recall that this class is the type for the partition $\bbP$. It has the following instance attributes. 
\begin{itemize}
	\item \mbox{$\compoundComponents: \PComponentList$} is a doubly linked list consisting of compound components of $\bbP$. 
	\item \mbox{$\simpleComponents: \PComponentList$} is a doubly linked list consisting of simple components of $\bbP$. 
\end{itemize}

The constructor $\PPartition()$ initializes the above mentioned attributes to newly created empty lists. The class has the method $\addPComponent(pc)$, which adds the $\bbP$-component $pc$ to the list $\compoundComponents$ or $\simpleComponents$ depending on whether $pc$ is compound or not.

\bigskip

\noindent
{\bf PComponentEdge}. Recall that this class is the type for objects specifying information about edges connecting a vertex $x$ to (vertices belonging to) a component $Y$ of the partition $\bbP$. It has the following instance attributes. 
\begin{itemize}
\item $\counter$ is a map of type $\SE \to ((0,1] \to \NN)$. If $r \in \SE$ is a key of $\counter$ and $d$ is a key of the map $\counter[r]$, then $\counter[r][d]$ is the number of vertices $y$ of $Y$ such that $E(x, r, y) = d$.
\item \mbox{$\departingPComponentEdge: \PComponentEdge$}. When the component $Y$ of $\bbP$ is going to be replaced by $X$ and $Y \setminus X$, the current $\bbP$-component-edge changes to a $\bbP$-component-edge with the destination $Y \setminus X$, a new $\bbP$-component-edge with the destination $X$ is created, and the attribute $\departingPComponentEdge$
of the current $\bbP$-component-edge is set to that new $\bbP$-component-edge.

\item \mbox{$\sourcePComponentEdge: \PComponentEdge$}. This attribute is a converse of $\departingPComponentEdge$. That is, the current $\bbP$-component-edge is equal to the attribute $\departingPComponentEdge$ of the object $\sourcePComponentEdge$ if they are set.
\end{itemize}
The class $\PComponentEdge$ also has the following methods.
\begin{itemize}
\item $\pushKey(r,d)$: This method increases the value of $\counter[r][d]$ by 1. (If $r$ is not a key of $\counter$, then $\counter[r]$ is first set to an empty map of the type $(0,1] \to \NN$. Next, if $d$ is not a key of $\counter[r]$, then $\counter[r][d]$ is set to 1, without a further increment.)
\item $\popKey(r,d)$: This method decreases the value of $\counter[r][d]$ by 1, under the assumption $r$ is a key of $\counter$ and $d$ is a key of $\counter[r]$. If $\counter[r][d]$ becomes 0, then the key $d$ is deleted from
the map $\counter[r]$. If $\counter[r]$ becomes an empty map, then the key $r$ is deleted from the map $\counter$.

\item $\maxKey(r)$: This method returns the biggest key of the map $\counter[r]$ if $r$ is a key of the map $\counter$ and the map $\counter[r]$ is not empty, and returns 0 otherwise.
\end{itemize}

The default constructor $\PComponentEdge()$ sets $\counter$ to a newly created empty map and sets the additional attributes to $\Null$. The constructor $\PComponentEdge(spce)$ differs from the default in that it also sets $\sourcePComponentEdge$ to $spce$. 

\bigskip

\noindent
{\bf LabelDegree}. Recall that this class is the type for objects representing the degree of a vertex label or an edge together with information about that vertex label or edge. We intend to use a vector of objects of this type in order to process all vertices' labels and edges in the ascending order w.r.t.\ their degrees. The class has the following instance attributes. 
\begin{itemize}
\item $\degree: (0,1]$ is the fuzzy degree of the involved vertex label or edge.
\item $\aboutVertex: \bool$ specifies whether the current object is about a vertex label or an edge.
\item $\vertexLabel: \SV$ is available only when $\aboutVertex = \true$.
\item $\vertex: \Vertex$ is available only when $\aboutVertex = \true$.
\item $\edge: \Edge$  is available only when $\aboutVertex = \false$.
\end{itemize}

The constructor $\LabelDegree(v,p)$ sets $\vertex$ and $\vertexLabel$ to $v$ and $p$, respectively, sets $\aboutVertex$ to $\true$, $\degree$ to $v.\vLabel[p]$, and $\edge$ to $\Null$. 
The constructor $\LabelDegree(e)$ sets $\edge$ to $e$, $\aboutVertex$ to $\false$, $\degree$ to $e.\degree$, $\vertex$ and $\vertexLabel$ to $\Null$. 


\subsection{Initialization}
\label{sec: init}

Our revision of Algorithm~\ref{algCompFP} uses the procedure $\Initialize()$ (on page~\pageref{procInitialize}), which sets up the global variables $\bbP$, $\bbB$, $\vertices$, $\labelDegrees$, $\labelDegreesIdx$ and $\allDegrees$. The vector $\labelDegrees$ contains objects of the type $\LabelDegree$ about all vertices' labels and edges in ascending order w.r.t.\ their degrees. We use the variable $\labelDegreesIdx$, which is initialized to $0$, to keep track of the beginning of the remaining part of the vector $\labelDegrees$ to be processed. The vector $\allDegrees$ corresponds to the sequence $d_0,d_1,\ldots,d_k$ of values used in Algorithm~\ref{algCompFP}. 

\begin{procedure}[t]
\caption{Initialize()\label{procInitialize}} 
	
construct a vector \mbox{$\vertices: \Vector(\Vertex)$} and a doubly linked list \mbox{$\vertices_2: \VertexList$} that contain all the vertices of $G$\label{step: UDWLS 1}\; 
construct a vector \mbox{$\edges: \Vector(\Edge)$} that contains all the edges of $G$ by calling the static method $\Vertex.\getVertex(id)$ and the constructor $\Edge(r,x,y,d,\Null)$ appropriately (this also sets up the lists of coming edges for the vertices)\label{step: UDWLS 2}\;

\BlankLine

$\bbP := \New  \PPartition()$\label{step: UDWLS 3}\;
$pc := \New \PComponent(\bbP)$\;
$\bbB := \New \SCBlock(\Null,\vertices_2, pc)$\label{step: UDWLS 4}\;

\BlankLine

create an empty map $\pcEdges: \Vertex \to \PComponentEdge$\label{step: UDWLS 5}\; 
\ForEach{$x \in \vertices$\label{step: UDWLS 6}}{
	$\pcEdges[x] := \New \PComponentEdge()$\label{step: UDWLS 7}\; 
}
\ForEach{$e \in \edges$\label{step: UDWLS 8}}{
	$e.\pcEdge := \pcEdges[e.\edgeOrigin]$\label{step: UDWLS 9}\;
	$e.\pcEdge.\pushKey(e.\edgeLabel,e.\degree)$\label{step: UDWLS 10}\;
}

\BlankLine

create an empty vector $\labelDegrees:\Vector(\LabelDegree)$\label{step: UDWLS 11}\;
\ForEach{$x \in \vertices$ and $p \in x.\vLabel.\keys()$\label{step: UDWLS 12}}{
	$\labelDegrees.\add(\New \LabelDegree(x,p))$\label{step: UDWLS 13}\;
}
\ForEach{$e \in \edges$\label{step: UDWLS 14}}{
	$\labelDegrees.\add(\New \LabelDegree(e))$\label{step: UDWLS 15}\;
}
sort $\labelDegrees$ by the attribute $\degree$ in ascending order\label{step: UDWLS 16}\;
$\labelDegreesIdx := 0$\label{step: UDWLS 17}\tcp*{the current index of the vector $\labelDegrees$}
create a vector $\allDegrees$ that contains 0 as the first element and all distinct values of the attribute $\degree$ of the elements of $\labelDegrees$ in ascending order\label{step: UDWLS 18}\;
\end{procedure}

Let's analyze the complexity of the procedure $\Initialize()$. 
Recall that the sizes of $\SE$ and $\SV$ are assumed to be bounded by a constant. 
The time needed for running the steps is as follows: 
\ref{step: UDWLS 1}:~$O(n\log{n})$;  
\ref{step: UDWLS 2}:~$O(m\log{n})$;  
\mbox{\ref{step: UDWLS 3}-\ref{step: UDWLS 5}}:~$O(1)$;   
\mbox{\ref{step: UDWLS 6}-\ref{step: UDWLS 7}}:~$O(n \log{n})$;   
\mbox{\ref{step: UDWLS 8}-\ref{step: UDWLS 10}}:~$O(m \log{n})$;   
\mbox{\ref{step: UDWLS 11}-\ref{step: UDWLS 15}}:~$O(m+n)$;   
\ref{step: UDWLS 16}:~$O((m+n) \log{n})$; 
and \mbox{\ref{step: UDWLS 17}-\ref{step: UDWLS 18}}:~$O(m+n)$.   
Thus, the time complexity of the procedure $\Initialize()$ is of order \mbox{$O((m+n)\log{n})$}. 

\subsection{The Revised Algorithm}

We revise Algorithm~\ref{algCompFP} to obtain Algorithm~\ref{algCompFPt} (on page~\pageref{algCompFPt}), which uses the classes specified in Section~\ref{sec: DS} and the procedure $\Initialize()$ given in Section~\ref{sec: init}. The revised algorithm and its subroutines use the global variables set up by the initialization. 
The statement~\ref{step: HGFAJ 4b} ($\refineA(\bbB, d_i, \bbP)$) of Algorithm~\ref{algCompFP} is simulated by the statements \ref{step: JHKXS 5}-\ref{step: JHKXS 6} (if $d_i = 0$ then $\RefineAA()$ else $\RefineAB(d_i)$) of Algorithm~\ref{algCompFPt}. 
The statements \ref{step: HGFAJ 7}-\ref{step: HGFAJ 8} ($\refineB(\bbB, d_i, X, Y)$; refine $\bbP$ by replacing $Y$ with $X$ and $Y \setminus X$) of Algorithm~\ref{algCompFP} are simulated by the statement~\ref{step: JHKXS 10} ($\RefineB(d_i, X)$) of Algorithm~\ref{algCompFPt}. 

\begin{algorithm}[t]
\caption{\CompFPt\label{algCompFPt}}
\Input{a finite fuzzy graph $G = \tuple{V, E, L, \SV, \SE}$.}
\Output{the fuzzy partition corresponding to the greatest fuzzy bisimulation of~$G$.}
	
\smallskip

$\Initialize()$\;
$k := \allDegrees.\length - 1$\label{step: JHKXS 2}\;
	
\ForEach{$i$ from 0 to $k-1$}{
	$d_i := \allDegrees[i]$\label{step: JHKXS 4}\;
	
	\BlankLine
	
	\lIf{$d_i = 0$}{$\RefineAA()$\label{step: JHKXS 5}}
	\lElse{$\RefineAB(d_i)$\label{step: JHKXS 6}}
	
	\BlankLine

	\While{not $\bbP.\compoundComponents.\Empty()$}{
		$Y := \bbP.\compoundComponents.\first()$\label{step: JHKXS 8}\;
		$X := Y.\smallerBlock()$\;
		$\RefineB(d_i, X)$\label{step: JHKXS 10}\;
	}	
}

\ForEach{$Y \in \bbP.\simpleComponents$}{
	$Y.\scBlocks.\first().\crisp := \true$\;
}	
\Return $\bbB$.	
\end{algorithm}


\begin{procedure}[t]
	\caption{Refine$_{1a}$()\label{proc: RefineAA}}
		\tcp{$\bbB$ is a simple block}
		\ForEach{$ld \in \labelDegrees$\label{RefineAA 1}}{
			\uIf{$ld.\aboutVertex$}{
				$x := ld.\vertex$, $p := ld.\vertexLabel$\;
				$x.\vLabelB.\add(p)$\;
			}
			\Else{
				$e := ld.\edge$, $x := e.\edgeOrigin$, $r := e.\edgeLabel$\;
				$x.\vLabelC.\add((r,\bbB.\pComponent))$\label{RefineAA 7}\;
			}
		}
		\ForEach{$x \in \vertices$\label{RefineAA 8}}{
			$key := (x.\vLabelB,x.\vLabelC)$\;
			\If{$key \notin \bbB.\departingBlocks_1.\keys()$}{set $\bbB.\departingBlocks_1[key]$ to a newly created empty list\;}
			move $x$ from $\bbB.\elements$ to $\bbB.\departingBlocks_1[key]$\label{RefineAA 11}\;
		}
		
		\uIf{$\bbB.\departingBlocks_1.\keys().\length = 1$\label{RefineAA 12}}{
			let $key$ be the unique element of $\bbB.\departingBlocks_1.\keys()$\;
			swap $\bbB.\departingBlocks_1[key]$ and $\bbB.\elements$\label{RefineAA 15}\;
			$\bbB.\departingBlocks_1.\clear()$\label{RefineAA 15b}\;
		}
		\Else{\label{RefineAA 16}
			$pc := \bbB.\pComponent$\;
			$pc.\removeBlock(\bbB)$\;
			$\bbB_2 := \New \FBlock(\Null,0)$\;
			\ForEach{$key \in \bbB.\departingBlocks_1.\keys()$}{
				$\New \SCBlock(\bbB_2, \bbB.\departingBlocks_1[key], pc)$\;
			}
			$\bbB := \bbB_2$\label{RefineAA 22}\;
		}
		\ForEach{$x \in \vertices$\label{RefineAA 23}}{
			$x.\vLabelB.\clear()$\;
			$x.\vLabelC.\clear()$\label{RefineAA 25}\;
		}
\end{procedure}

The procedure $\RefineAA()$, defined on page~\pageref{proc: RefineAA}, deals with the case where $\bbB$ is a simple block and $\bbP$ consists of only one component ($\bbB.\pComponent$). Its aim is to refine $\bbB$ to make it $0$-cut stable w.r.t.~$\bbP$. By definition, $\bbB$ should be split so that, if $x$ and $x'$ are vertices of $G$, then they belong to the same simple block after the splitting iff:
\begin{itemize}
\item for every $p \in \SV$, $p \in x.\vLabel.\keys()$ iff $p \in x'.\vLabel.\keys()$; and
\item for every $r \in \SE$, there exists an edge $e$ with $e.\edgeOrigin = x$ iff there exists an edge $e'$ with $e'.\edgeOrigin = x'$. 
\end{itemize}
Note that, at this stage, for every egde $e$ of $G$, $e.\pcEdge = \bbB.\pComponent$. The new block to put a vertex $x$ in is identified by the following two sets:
\begin{itemize}
	\item $\{p \mid$ there exists $ld \in \labelDegrees$ such that $ld.\aboutVertex = \true$, $ld.\vertex = x$ and $ld.\vertexLabel = p\}$, which is computed and stored in $x.\vLabelB$;\footnote{This set corresponds to $x.\vLabel.\keys()$.} and 
	\item $\{(r,\bbB.\pComponent) \mid$ there exists $ld \in \labelDegrees$ such that $ld.\aboutVertex = \false$ and, for $e = ld.\edge$, $e.\edgeOrigin = x$ and $e.\edgeLabel = r\}$, which is computed and stored in $x.\vLabelC$.
\end{itemize}
The attributes $x.\vLabelB$ and $x.\vLabelC$ for vertices $x$ of $G$ are computed by the statements \ref{RefineAA 1}-\ref{RefineAA 7} of the procedure $\RefineAA()$. Preparing subblocks to split $\bbB$ into is done by the loop in the statements \ref{RefineAA 8}-\ref{RefineAA 11}. The splitting itself is done by the statements \ref{RefineAA 12}-\ref{RefineAA 22}. The statements \ref{RefineAA 12}-\ref{RefineAA 15b} deal with the case where the number of subblocks to split $\bbB$ into is equal to~1. For this case, they just restore $\bbB$ to the state before executing the procedure. The statements \ref{RefineAA 16}-\ref{RefineAA 22} deal with the other case. They replace $\bbB$ with a fuzzy block $\bbB_2$, where $\bbB_2.\degree = 0$ and the list $\bbB_2.\subblocks$ consists of simple blocks whose contents were computed by the mentioned statements \ref{RefineAA 8}-\ref{RefineAA 11}. The contents of the unique component of $\bbP$ are also updated appropriately. The statements \ref{RefineAA 23}-\ref{RefineAA 25} clear the auxiliary sets $x.\vLabelB$ and $x.\vLabelC$ for all vertices~$x$ of~$G$.

\newcommand{\verticesTBP}{\mathit{vertices\!\_tbp}}
\begin{procedure}
\small
	\caption{Refine$_{1b}$($d_i$)\label{proc: RefineAB}}
	\tcp{we have $d_i > 0$}
	set $\verticesTBP$ to an empty vector\label{RefineAB 1}\tcp*{it stands for vertices-to-be-processed}
	\While{$\labelDegreesIdx < \labelDegrees.\length$ and $\labelDegrees[\labelDegreesIdx].\degree = d_i$\label{RefineAB 2}}{
		$ld := \labelDegrees[\labelDegreesIdx]$, $\labelDegreesIdx := \labelDegreesIdx + 1$\label{RefineAB 3}\;
		\uIf{$ld.\aboutVertex$}{
			$x := ld.\vertex$, $p := ld.\vertexLabel$\;
			$x.\vLabelB.\add(p)$\;
			\If{not $x.\processed$}{
				$\verticesTBP.\add(x)$\;
				$x.\processed := \true$\;
			}
		}
		\Else{
			$e := ld.\edge$, $r := e.\edgeLabel$, $pce := e.\pcEdge$\; 
			\If{$pce.\maxKey(r) = d_i$}{
				$x := e.\edgeOrigin$, $pc := e.\edgeDest.\vertexBlock.\pComponent$\;
				$x.\vLabelC.\add((r,pc))$\label{RefineAB 14}\;
				\If{not $x.\processed$}{
					$\verticesTBP.\add(x)$\;
					$x.\processed := \true$\label{RefineAB 17}\;
				}
			}
		}
	}
	\ForEach{$x \in \verticesTBP$\label{RefineAB 18}}{
		$bx := x.\vertexBlock$, $key := (x.\vLabelB,x.\vLabelC)$\;
		\If{$key \notin bx.\departingBlocks_1.\keys()$}{set $bx.\departingBlocks_1[key]$ to a newly created empty list\;}
		move $x$ from $bx.\elements$ to $bx.\departingBlocks_1[key]$\label{RefineAB 22}\;
	}
	
	\ForEach{$x \in \verticesTBP$\label{RefineAB 23}}{
		$bx := x.\vertexBlock$\;
		\If{not $bx.\departingBlocks_1.\Empty()$}{
			\uIf{$bx.\departingBlocks_1.\keys().\length = 1$ and $bx.\elements.\Empty()$\label{RefineAB 26}}{
				let $key$ be the unique element of $bx.\departingBlocks_1.\keys()$\;
				swap $bx.\departingBlocks_1[key]$ and $bx.\elements$\label{RefineAB 29}\;
			}
			\Else{\label{RefineAB 30}
				$pc := bx.\pComponent$\;
				$bp := bx.\parent$\;
				\lIf{$bp \neq \Null$}{remove $bx$ from $bp.\subblocks$}
				$bx_2 := \New \FBlock(bp,d_i)$\;
				\lIf{$\bbB = bx$}{$\bbB := bx_2$}
				\ForEach{$key \in bx.\departingBlocks_1.\keys()$\label{RefineAB 33}}{
					$\New \SCBlock(bx_2, bx.\departingBlocks_1[key], pc)$\label{RefineAB 34}\;
				}
				\uIf{not $bx.\elements.\Empty()$}{
					add $bx$ to $bx_2.\subblocks$\;
					$bx.\parent := bx_2$\;
				}
				\lElse{$pc.\removeBlock(bx)$\label{RefineAB 39}}
			}	
			$bx.\departingBlocks_1.\clear()$\label{RefineAB 39b}\;
		}
	}
	
	\ForEach{$x \in \verticesTBP$\label{RefineAB 40}}{
		$x.\vLabelB.\clear()$,\ 
		$x.\vLabelC.\clear()$,\  
		$x.\processed := \false$\label{RefineAB 43}\;
	}
\end{procedure}

Consider the procedure $\RefineAB(d_i)$ (on page~\pageref{proc: RefineAB}) under the assumption that $d_i > 0$. Its aim is to refine $\bbB$ to make it $d_i$-cut stable w.r.t.~$\bbP$. Let $i$ be the index used in the statement~\ref{step: JHKXS 4} of Algorithm~\ref{algCompFPt} and let $d_{i-1} = \allDegrees[i-1]$. As Algorithm~\ref{algCompFPt} is designed to simulate Algorithm~\ref{algCompFP}, by Lemma~\ref{lemma: DJSJA}, we can assume that $\bbB$ is already $d_{i-1}$-cut stable w.r.t.~$\bbP$. Under this assumption, $\bbB$ should be refined so that, if $x$ and $x'$ belong to the same block of the leaf partition of $\bbB$, then they will still belong to the same block of the leaf partition of $\bbB$ after the refinement iff: 
\begin{itemize}
\item for every $p \in \SV$, $x.\vLabel[p] = d_i$ iff $x'.\vLabel[p] = d_i$; and
\item for every $r \in \SE$ and every $\bbP$-component $pc$, there exists an edge $e$ with $e.\edgeOrigin = x$, $e.\edgeLabel = r$, $e.\edgeDest.\vertexBlock.\pComponent = pc$ and $e.\pcEdge.\maxKey(r) = d_i$ iff there exists an edge $e'$ with $e'.\edgeOrigin = x'$, $e'.\edgeLabel = r$, $e'.\edgeDest.\vertexBlock.\pComponent = pc$ and $e'.\pcEdge.\maxKey(r) = d_i$.
\end{itemize}
Thus, the new block to put a vertex $x$ in is identified by the following two sets:
\begin{itemize}
\item $\{p \mid$ there exists $ld \in \labelDegrees$ such that $ld.\degree = d_i$, $ld.\aboutVertex = \true$, $ld.\vertex = x$ and $ld.\vertexLabel = p\}$, which is computed and stored in $x.\vLabelB$; and 
\item $\{(r,pc) \mid$ there exists $ld \in \labelDegrees$ such that $ld.\degree = d_i$, $ld.\aboutVertex = \false$ and, for $e = ld.\edge$, $e.\edgeOrigin = x$, $e.\edgeLabel = r$, $e.\edgeDest.\vertexBlock.\pComponent = pc$ and $e.\pcEdge.\maxKey(r) = d_i\}$, which is computed and stored in $x.\vLabelC$.
\end{itemize}
The attributes $x.\vLabelB$ and $x.\vLabelC$ for vertices $x$ of $G$ are computed by the ``while'' loop in the statements \ref{RefineAB 2}-\ref{RefineAB 17} of the procedure $\RefineAB(d_i)$. This loop also computes a vector $\verticesTBP$ (vertices to be processed) consisting of the vertices $x$ such that either $x.\vLabelB$ or $x.\vLabelC$ is not empty. Each block of the leaf partition of $\bbB$ may need to be split. Preparing subblocks to split blocks of the leaf partition of $\bbB$ into is done by the loop in the statements \ref{RefineAB 18}-\ref{RefineAB 22}, which sets up the maps $\departingBlocks_1$ for blocks of the leaf partition of $\bbB$. The splitting itself is done by the statements \ref{RefineAB 23}-\ref{RefineAB 39b}. For each vertex $x$ from $\verticesTBP$ and for the simple block $bx$ that contains $x$, if the map $bx.\departingBlocks_1$ is not empty, then $bx$ need to be processed for the splitting. In that case, splitting $bx$ is done by the statements \ref{RefineAB 26}-\ref{RefineAB 39b}. The statements \ref{RefineAB 26}-\ref{RefineAB 29} and \ref{RefineAB 39b} deal with the case where the number of subblocks to split $bx$ into is equal to~1. For this case, they just restore $bx$ to the state before executing the procedure. The statements \ref{RefineAB 30}-\ref{RefineAB 39b} deal with the other case. They replace $bx$ in the representation tree of $\bbB$ with a new fuzzy block $bx_2$, where $bx_2.\degree = d_i$ and the list $bx_2.\subblocks$ consists of simple blocks whose contents are either a key's value of the map $bx.\departingBlocks_1$ or $bx.\elements$ (if this list is not empty). The contents of the components of $\bbP$ are also updated appropriately. The statements \ref{RefineAB 40}-\ref{RefineAB 43} clear the sets $x.\vLabelB$ and $x.\vLabelC$ and the flag $x.\processed$ for all vertices~$x$ from~$\verticesTBP$.


\begin{figure*}
\begin{procedure}[H]
\caption{Refine$_2$($d_i,X$)\label{proc: RefineB}}
create a vector $\verticesX$ consisting of the elements of the list $X.\elements$\;
$\ComputePComponentEdges(\verticesX)$\label{RefineB 2}\;
$\ComputeSubblocks(d_i,\verticesX)$\label{RefineB 3}\;
$\DoSplitting(d_i,X,\verticesX)$\label{RefineB 4}\;
$\ClearAuxiliaryInfo(\verticesX)$\label{RefineB 5}\;
\end{procedure}

\smallskip

\begin{procedure}[H]
\caption{ComputePComponentEdges($\verticesX$)\label{proc: ComputePComponentEdges}}
\ForEach{$x \in \verticesX$ and $e \in x.\comingEdges$}{
	$pce := e.\pcEdge$\;
	\If{$pce.\departingPComponentEdge = \Null$}{
		$pce.\departingPComponentEdge = \New \PComponentEdge(pce)$\;
	}
	$dpce = pce.\departingPComponentEdge$\;
	$r := e.\edgeLabel$, $d := e.\degree$\;
	$pce.\popKey(r,d)$, $dpce.\pushKey(r,d)$\;
}
\end{procedure}

\smallskip

\begin{procedure}[H]
\caption{ComputeSubblocks($d_i,\verticesX$)\label{proc: ComputeSubblocks}}
\ForEach{$x \in \verticesX$ and $e \in x.\comingEdges$\label{step: JHFJQ 1}}{
	$v := e.\edgeOrigin$, $r := e.\edgeLabel$, $pce := e.\pcEdge$, $dpce := pce.\departingPComponentEdge$\;
	\lIf{$pce.\maxKey(r) > d_i$}{$v.\vLabelD.add(r)$}
	\lIf{$dpce.\maxKey(r) > d_i$}{$v.\vLabelE.add(r)$\label{step: JHFJQ 4}}
}

\ForEach{$x \in \verticesX$ and $e \in x.\comingEdges$\label{step: JHFJQ 5}}{
	$v := e.\edgeOrigin$, $bv := v.\vertexBlock$\;
	\If{not $v.\processed$}{
		$key := (v.\vLabelD,v.\vLabelE)$\;
		\If{$key \notin bv.\departingBlocks_2.\keys()$}{
			$bv.\departingBlocks_2[key] := \New \VertexList()$\;
		}
		move $v$ from $bv.\elements$ to $bv.\departingBlocks_2[key]$\;
		$v.\processed := \true$\label{step: JHFJQ 13}\;
	}
}
\end{procedure}
\end{figure*}

\begin{figure*}
\begin{procedure}[H]
\caption{DoSplitting($d_i,X,\verticesX$)\label{proc: DoSplitting}}
$Y := X.\pComponent$\label{step: JRMBA 1}\;
$Y.\removeBlock(X)$\;
$\createPComponent(Y.\pPartition, X)$\label{step: JRMBA 3}\;

\ForEach{$x \in \verticesX$ and $e \in x.\comingEdges$\label{step: JRMBA 4}}{
	$e.\pcEdge := e.\pcEdge.\departingPComponentEdge$\;
	$v := e.\edgeOrigin$, $bv := v.\vertexBlock$\;
	\If{not $bv.\departingBlocks_2.\Empty()$\label{step: JRMBA 7}}{
		\uIf{$bv.\departingBlocks_2.\keys().\length = 1$ and $bv.\elements.\Empty()$\label{step: JRMBA 8}}{
			let $key$ be the unique element of $bx.\departingBlocks_2.\keys()$\;
			swap $bv.\departingBlocks_2[key]$ and $bv.\elements$\label{step: JRMBA 11}\;
		}
		\Else{\label{step: JRMBA 12}
			$pc := bv.\pComponent$\;
			$bp := bv.\parent$\;
			\uIf{$bp = \Null$ or $bp.\degree < d_i$\label{step: JRMBA 16}}{
				\lIf{$bp \neq \Null$}{remove $bv$ from $bp.\subblocks$}
				$bv_2 := \New \FBlock(bp,d_i)$\;
				\lIf{$\bbB = bv$}{$\bbB := bv_2$}
				\ForEach{$key \in bv.\departingBlocks_2.\keys()$}{
					$\New \SCBlock(bv_2, bv.\departingBlocks_2[key], pc)$\;
				}
				\uIf{not $bv.\elements.\Empty()$}{
					add $bv$ to $bv_2.\subblocks$\;
					$bv.\parent := bv_2$\;
				}
				\lElse{$pc.\removeBlock(bv)$\label{step: JRMBA 23}}
			}
			\Else{\label{step: JRMBA 24}
				\ForEach{$key \in bv.\departingBlocks_2.\keys()$}{
					$\New \SCBlock(bp, bv.\departingBlocks_2[key], pc)$\;
				}
				\If{$bv.\elements.\Empty()$}{
					remove $bv$ from $bp.\subblocks$\;	
					$pc.\removeBlock(bv)$\label{step: JRMBA 29}\;	
				}
			}
		}
		$bv.\departingBlocks_2.\clear()$\label{step: JRMBA 29b}\;
	}
}
\end{procedure}

\smallskip

\begin{procedure}[H]
\caption{ClearAuxiliaryInfo($\verticesX$)\label{proc: ClearAuxiliaryInfo}}
	
\ForEach{$x \in \verticesX$ and $e \in x.\comingEdges$}{
	$v := e.\edgeOrigin$\;
	$v.\vLabelD.\clear()$, 
	$v.\vLabelE.\clear()$, 
	$v.\processed := \false$\; 
	$pce := e.\pcEdge$, $spce := pce.\sourcePComponentEdge$\;
	\If{$spce \neq \Null$}{
		$spce.\departingPComponentEdge := \Null$\;
		$pce.\sourcePComponentEdge := \Null$\;
	}
}
\end{procedure}
\end{figure*}

Consider the procedure $\RefineB(d_i, X)$ given on page~\pageref{proc: RefineB}, with subroutines defined on pages \pageref{proc: ComputePComponentEdges} and \pageref{proc: DoSplitting}. Recall that it is used in the statement~\ref{step: JHKXS 10} of Algorithm~\ref{algCompFPt}. The parameter $X$ is a simple block of the leaf partition of $\bbB$. 
Let $i$ be the index used in the statement~\ref{step: JHKXS 4} of Algorithm~\ref{algCompFPt} and let $d_j = \allDegrees[j]$ for $0 \leq j < i$. Let $Y = X.\pComponent$ (the $\bbP$-component that contains the simple block $X$). It is the same as the variable $Y$ in the statement~\ref{step: JHKXS 8} of Algorithm~\ref{algCompFPt} before the call $\RefineB(d_i, X)$. The aim of the procedure is to refine $\bbB$ to make it $d_i$-cut stable w.r.t.\ both $X$ and $Y \setminus X$, and to refine $\bbP$ by replacing its component $Y$ with $X$ and $Y \setminus X$, treating $X$ and $Y$ as sets of vertices. The procedure $\RefineB(d_i, X)$ starts with creating a vector $\verticesX$ consisting of the elements of the list $X.\elements$ (we need a copy because $X$ may be split by the procedure). Then, it executes four subroutines, which are discussed below. 

The call $\ComputePComponentEdges(\verticesX)$ in the procedure $\RefineB(d_i, X)$ prepares $\bbP$-component-edges that will connect vertices to two future $\bbP$-components, whose contents are $Y\setminus X$ or $X$, respectively, treating $X$ and $Y$ as sets of vertices. 
Each $\bbP$-component-edge that connects a vertex $v$ to the $\bbP$-component $Y$ via edges with the destination belonging to $\verticesX$ is updated to become the one that
\begin{itemize}
\item plays the role of a $\bbP$-component-edge connecting $v$ to the future $\bbP$-component with contents $Y \setminus X$, and 
\item has the attribute $\departingPComponentEdge$ set to a newly created $\bbP$-component-edge intended for connecting $v$ to the future $\bbP$-component with contents $X$. 
\end{itemize}

Let $v$ be an arbitrary vertex of $G$ and $pce$ the $\bbP$-component-edge connecting $v$ to $Y$ (i.e., the one such that there exists an edge $e$ connecting $v$ to a vertex $x$ such that $e.\pcEdge = pce$ and $x.\vertexBlock.\pComponent = Y$). Let $dpce = pce.\departingPComponentEdge$. After executing $\ComputePComponentEdges(\verticesX)$ in the statement~\ref{RefineB 2} of the procedure $\RefineB(d_i, X)$, we have that, for every $r \in \SE$, \mbox{$\sup E(v,r,Y \setminus X) =$} $pce.\maxKey(r)$, $\sup E(v,r,X) = dpce.\maxKey(r)$ if $dpce \neq \Null$, and $\sup E(v,r,X) = 0$ otherwise, treating $X$ and $Y$ as sets of vertices. 
Let $v'$ be another vertex belonging to the same simple block of the leaf partition of $\bbB$ as $v$ (i.e., $v'.\vertexBlock = v.\vertexBlock$). Let $pce'$ be the $\bbP$-component-edge connecting $v'$ to $Y$ and let $dpce' = pce'.\departingPComponentEdge$. As Algorithm~\ref{algCompFPt} is designed to simulate Algorithm~\ref{algCompFP}, by Lemma~\ref{lemma: DJSJA}, we can assume that before calling $\RefineB(d_i, X)$ in the statement~\ref{step: JHKXS 10} of Algorithm~\ref{algCompFPt} $\bbB$ is already $d_j$-cut stable, for all $0 \leq j < i$, and $d_i$-cut stable w.r.t.~$\bbP$. Thus, for every $p \in \SV$, $L(v)(p) > d_i$ iff $L(v')(p) > d_i$. Furthermore, $\bbB$ should be refined by the procedure $\RefineB(d_i, X)$ so that $v$ and $v'$ will still belong to the same block of the leaf partition of $\bbB$ after the refinement iff either $dpce = \Null$ and $dpce' = \Null$ or the following two conditions hold: 
\begin{itemize}
\item for every $r \in \SE$, $pce.\maxKey(r) > d_i$ iff $pce'.\maxKey(r) > d_i$;
\item $dpce \neq \Null$, $dpce' \neq \Null$ and, for every $r \in \SE$, (\mbox{$dpce.\maxKey(r) > d_i$} iff \mbox{$dpce'.\maxKey(r) > d_i$}).
\end{itemize}
The first condition means that, for every $r \in \SE$, \mbox{$\sup E(v,r,Y \setminus X) > d_i$} iff \mbox{$\sup E(v',r,Y \setminus X) > d_i$}, whereas the second condition means that $dpce \neq \Null$, $dpce' \neq \Null$ and, for every $r \in \SE$, (\mbox{$\sup E(v,r,X) > d_i$} iff \mbox{$\sup E(v',r,X) > d_i$}). 
Therefore, the new block to put $v$ in is identified either by $dpce = \Null$ or by the following two sets:
\begin{itemize}
\item $\{r \in \SE \mid pce.\maxKey(r) > d_i\}$, which is computed and stored in $v.\vLabelD$;
\item $\{r \in \SE \mid dpce.\maxKey(r) > d_i\}$, which is computed and stored in $v.\vLabelE$. 
\end{itemize}

The call $\ComputeSubblocks(d_i,\verticesX)$ in the statement~\ref{RefineB 3} of the procedure $\RefineB(d_i,X)$ prepares subblocks to split blocks of the leaf partition of B into. It is defined on page~\pageref{proc: ComputePComponentEdges}. The loop in the statements \ref{step: JHFJQ 1}-\ref{step: JHFJQ 4} computes the sets $v.\vLabelD$ and $v.\vLabelE$ mentioned above for all vertices $v$ such that there are edges connecting $v$ to a vertex belonging to $\verticesX$. After that, the next loop moves such a vertex $v$ from $bv.\elements$ to $bv.\departingBlocks_2[(v.\vLabelD,v.\vLabelE)]$, where $bv = v.\vertexBlock$.

The above discussed calls of $\ComputePComponentEdges$ and $\ComputeSubblocks$ are only a preparation of the procedure $\RefineB(d_i, X)$ for the refinement. The splitting is really done by the call $\DoSplitting(d_i,X,\verticesX)$ in the statement~\ref{RefineB 4}. This subroutine is defined on page~\pageref{proc: DoSplitting}. Its statements \ref{step: JRMBA 1}-\ref{step: JRMBA 3} remove the block $X$ from the $\bbP$-component $Y$ and add to $\bbP$ a new component consisting only of the block $X$. Splitting blocks of the leaf partition of $\bbB$ is done by the loop in the statements \ref{step: JRMBA 4}-\ref{step: JRMBA 29b}. For each vertex $x$ from $\verticesX$ and for each edge coming to $x$ from a vertex $v$, if the map $bv.\departingBlocks_2$ is not empty, where $bv = v.\vertexBlock$, then $bv$ need to be processed for the splitting. In that case, splitting $bv$ is done by the statements \ref{step: JRMBA 8}-\ref{step: JRMBA 29b}. The statements \ref{step: JRMBA 8}-\ref{step: JRMBA 11} and \ref{step: JRMBA 29b} deal with the case where the number of subblocks to split $bv$ into is equal to~1. For this case, they just restore $bv$ to the state before executing the procedure $\RefineB(d_i,X)$. The statements \ref{step: JRMBA 12}-\ref{step: JRMBA 29b} deal with the other case. Let $bp = bv.\parent$. As specified on page \pageref{place HDJHS} for the procedure $\refineB(\bbB, d, X, Y)$ used in Algorithm~\ref{algCompFP}, there are two subcases. If $bp = \Null$ or $bp.\degree < d_i$, then $bv$ is replaced in the representation tree of $\bbB$ by a new fuzzy block $bv_2$, where $bv_2.\degree = d_i$ and the list $bv_2.\subblocks$ consists of simple blocks whose contents are either a key's value of the map $bv.\departingBlocks_2$ or $bv.\elements$ (if this list is not empty). In the other case, the list $bp.\subblocks$ is modified by adding to it new simple blocks whose contents are a key's value of the map $bv.\departingBlocks_2$ and by removing $bv$ if $bv.\elements$ is empty. The contents of the component of $\bbP$ that contains $bv$ are also updated appropriately. 

The call $\ClearAuxiliaryInfo(\verticesX)$ in the statement~\ref{RefineB 5} of the procedure $\RefineB(d_i,X)$ clears the auxiliary attributes $\vLabelD$, $\vLabelE$ and $\processed$ of vertices and the auxiliary attributes $\departingPComponentEdge$ and $\sourcePComponentEdge$ of $\bbP$-component-edges. The subroutine is defined on page~\pageref{proc: ClearAuxiliaryInfo}. 

It is worth noticing the similar scheme of the procedures $\RefineAA$, $\RefineAB$ and $\RefineB$. By the description and discussion given above about these procedures, it can be seen that the statements \ref{step: JHKXS 5}-\ref{step: JHKXS 6} (if $d_i = 0$ then $\RefineAA()$ else $\RefineAB(d_i)$) of Algorithm~\ref{algCompFPt} strictly simulate the call $\refineA(\bbB, d_i, \bbP)$ in Algorithm~\ref{algCompFP}. Similarly, the call $\RefineB(d_i, X)$ in Algorithm~\ref{algCompFPt} strictly simulates the statements \ref{step: HGFAJ 7}-\ref{step: HGFAJ 8} ($\refineB(\bbB, d_i, X, Y)$; refine $\bbP$ by replacing $Y$ with $X$ and $Y \setminus X$) of Algorithm~\ref{algCompFP}. By Theorem~\ref{theorem: correctness 1}, we reach the following result.

\begin{theorem}\label{theorem: correctness 2}
Algorithm~\ref{algCompFPt} is correct. That is, it is a correct algorithm for computing the fuzzy partition corresponding to the greatest fuzzy bisimulation of a given finite fuzzy graph.
\end{theorem}

We have implemented Algorithm~\ref{algCompFPt} in Python and shared the codes publicly~\cite{CompFP-prog}. The codes are very similar to the pseudocodes given in this section. The user can experiment with them to display various information. 

\subsection{Complexity Analysis}
\label{sec: comp-anal}

Recall that $n = |V|$, $m = |\{\tuple{x,r,y} \in$ \mbox{$V \times \SE \times V:$} $E(x,r,y) > 0\}|$ and 
$l = |\{E(x,r,y)$ : $\tuple{x,r,y} \in$ \mbox{$V \times \SE \times V\}|$}. Assume that $l \geq 2$. Also recall that the sizes of $\SE$ and $\SV$ are assumed to be bounded by a constant. We now estimate the time complexity of Algorithm~\ref{algCompFPt} in terms of $n$, $m$ and $l$. 

Let $k$ be the variable set in the statement~\ref{step: JHKXS 2} of Algorithm~\ref{algCompFPt}. We have $k \leq l + n \cdot |\SV|$. 

Consider the procedure $\RefineAA()$. During its execution, there is only one $\bbP$-component and the sizes of $x.\vLabelB$ and $x.\vLabelC$, for $x \in \vertices$, are bounded by a constant. The time taken by the statements \ref{RefineAA 1}-\ref{RefineAA 7} (respectively, \ref{RefineAA 8}-\ref{RefineAA 25}) of this procedure is of order $O(m+n)$ (respectively, $O(n)$). Thus, the time taken by this procedure is of order $O(m+n)$. 

Consider the procedure $\RefineAB(d_i)$. Let $h_i$ be the number of elements of $\labelDegrees$ whose attribute $\degree$ has the value $d_i$. We have $\sum_{1 \leq i < k} h_i < m + n\cdot |\SV| = O(m+n)$. The time taken by the procedure $\RefineAB(d_i)$ is of order $O(h_i\cdot\log{n})$. To see this, observe that, in the loops ``foreach $x \in \verticesTBP$ do'' in the statements \ref{RefineAB 18}-\ref{RefineAB 22} and \ref{RefineAB 23}-\ref{RefineAB 39b}:
\begin{itemize}
\item the sum of the lengths of serialized representations of the keys of the maps $bx.\departingBlocks_1$ for $x \in \verticesTBP$ and $bx = x.\vertexBlock$ is of order $O(h_i)$;
\item the number of keys in each of the mentioned maps $bx.\departingBlocks_1$ is of order $O(n)$;
\item the total number of vertices belonging to the lists that are keys' values of the mentioned maps $bx.\departingBlocks_1$ is of order $O(h_i)$, and such lists are non-empty. 
\end{itemize}
Hence, the total time taken by all the calls $\RefineAB(d_i)$, for $1 \leq i < k$, is of order $O((m+n)\log{n})$. 

Given \mbox{$X:\SCBlock$}, we write $|X|$ to denote $|X.\elements|$ and write \mbox{$|\!\uparrow\!X|$} to denote the number of edges coming to a vertex of $X$. Given \mbox{$Y:\PComponent$}, we also write $|Y|$ to denote the total number of vertices that belong to a block of~$Y$. 

Consider the procedure $\RefineB(d_i,X)$. The call $\ComputePComponentEdges(\verticesX)$ runs in time \mbox{$O(|X| + |\!\uparrow\!X|\cdot \log{l})$}. All the calls $\ComputeSubblocks(d_i,\verticesX)$, $\DoSplitting(d_i,X,\verticesX)$ and $\ClearAuxiliaryInfo(\verticesX)$ run in time \mbox{$O(|X| + |\!\uparrow\!X|)$}. Therefore, the procedure $\RefineB(d_i,X)$ runs in time \mbox{$O(|X| + |\!\uparrow\!X|\cdot \log{l})$}. 
Dividing this cost for the individual vertices of $X$, we can assume that the cost assigned to each vertex $x \in X.\elements$ in a call $\RefineB(d_i,X)$ is of order \mbox{$O(1 + |\!\uparrow\!x|\cdot \log{l})$}, where \mbox{$|\!\uparrow\!x|$} is the number of edges coming to~$x$. 

Fix an arbitrary vertex $x \in V$. Let's estimate the number of calls $\RefineB(d_i,X)$ during the execution of Algorithm~\ref{algCompFPt} for $G$ such that $x$ is a vertex of $X$. Denote it by $f(x)$. Observe that, if $\RefineB(d_i,X)$ is such a call at some step, then the next call of $\RefineB$ with that property at some later step, if it exists, must be $\RefineB(d_{i'},X')$ with $|X'| \leq |X'.\pComponent|/2 \leq |X|/2$. Extending this understanding, we conclude that $f(x) \leq \log{n}$. 

Therefore, the total time taken by all the calls of $\RefineB$ in the statement~\ref{step: JHKXS 10} of the execution of Algorithm~\ref{algCompFPt} for $G$ is of order 
\[ O(\sum_{x \in V} \log{n} \cdot (1 + |\!\uparrow\!x|\cdot \log{l})), \]
which is of order $O((m\log{l} + n)\log{n})$. 

As discussed earlier in this subsection, the total time taken by the statements \ref{step: JHKXS 5}-\ref{step: JHKXS 6} of the Algorithm~\ref{algCompFPt} is of order $O((m+n)\log{n})$. As estimated in Section~\ref{sec: init}, the time taken by $\Initialize()$ is also of order \mbox{$O((m+n)\log{n})$}. Hence, we arrive at the following theorem. 

\begin{theorem}\label{theorem: DJAJD}
Algorithm~\ref{algCompFPt} has a time complexity of order $O((m\log{l} + n)\log{n})$.
\end{theorem}

If $l$ is bounded by a constant (e.g., when $l = 2$ and $G$ is a crisp graph), then the time complexity of Algorithm~\ref{algCompFPt} is of order $O((m+n)\log{n})$. If \mbox{$m \geq n$}, then, taking $l = n^2$ for the worst case, the complexity of the algorithm is of order $O(m\log^2(n))$. 


\section{Computing Fuzzy Bisimulations}
\label{sec: comp FB}

We present Algorithm~\ref{algCompFB} (on page~\pageref{algCompFB}) for computing the greatest fuzzy bisimulation between two finite fuzzy graphs $G_1 = \tuple{V_1, E_1, L_1, \SV, \SE}$ and $G_2 = \tuple{V_2, E_2, L_2, \SV, \SE}$ over the same signature $\tuple{\SV, \SE}$. It applies Algorithm~\ref{algCompFPt} to $G_1 \uplus G_2$, converts the resulting fuzzy partition to the corresponding fuzzy equivalence relation by using the procedure $\ConvertFPtoFB$, and then restricts the obtained result to $V_1 \times V_2$. 

\begin{algorithm}[t]
\caption{\CompFB\label{algCompFB}}
\Input{finite fuzzy graphs $G_1 = \tuple{V_1, E_1, L_1, \SV, \SE}$ and $G_2 = \tuple{V_2, E_2, L_2, \SV, \SE}$.}
\Output{the greatest fuzzy bisimulation between $G_1$ and $G_2$.}
let $n_1 = |V_1|$ and $n_2 = |V_2|$\;
rename the vertices of $G_1$ and $G_2$ so that $V_1 = 0..(n_1-1)$ and $V_2 = n_1..(n_1 + n_2 - 1)$, where $a..b$ means $\{a,a+1,\ldots,b\}$, and keep the information to restore the names later\label{step: algCompFB 2}\;
let $G = G_1 \uplus G_2$\label{step: algCompFB 3}\;
let $B$ be the result of applying Algorithm~\ref{algCompFPt} to $G$\label{step: algCompFB 4}\; 
$f := \ConvertFPtoFB(B)$\label{step: algCompFB 5}\;
let $f':V_1 \times V_2 \to [0,1]$ be the fuzzy relation obtained from $f$ by restricting it to $V_1 \times V_2$ and restoring the names of the vertices\label{step: algCompFB 6}\;
\Return $f'$\;
\end{algorithm}

\begin{theorem}
Algorithm~\ref{algCompFB} is correct. That is, it is a correct algorithm for computing the greatest fuzzy bisimulation between given finite fuzzy graphs $G_1$ and $G_2$. Its time complexity is of order $O(m\!\cdot\!\log{l}\!\cdot\!\log{n} + n^2)$, where $n = n_1 + n_2$, $m = m_1 + m_2$, $l = l_1 + l_2$, and $n_1$, $m_1$, $l_1$ (respectively, $n_2$, $m_2$, $l_2$) are the number of vertices, the number of non-zero edges and the number of different fuzzy degrees of edges of~$G_1$ (respectively,~$G_2$). 
\end{theorem}

\begin{proof}
The correctness of Algorithm~\ref{algCompFB} immediately follows from the correctness of Algorithm~\ref{algCompFPt} (Theorem~\ref{theorem: correctness 2}) and the procedure $\ConvertFPtoFB$ (Proposition~\ref{prop: RHJQK}). 
Concerning the complexity, observe that the time taken by the statements are as follows: \ref{step: algCompFB 2}: $O((m + n)\log{n})$, \ref{step: algCompFB 3}: $O(m + n)$, \ref{step: algCompFB 4}: $O((m \log{l} + n)\log{n})$ (by Theorem~\ref{theorem: DJAJD}), \ref{step: algCompFB 5}: $O(n^2)$ (by Proposition~\ref{prop: RHJQK}), \ref{step: algCompFB 6}: $O(n^2)$. Summing up, we conclude that Algorithm~\ref{algCompFB} runs in time $O(m\!\cdot\!\log{l}\!\cdot\!\log{n} + n^2)$. 
\myend
\end{proof}

If $l$ is bounded by a constant (e.g., when $l = 2$ and $G_1$ and $G_2$ are crisp graphs), then the time complexity of Algorithm~\ref{algCompFB} is of order $O(m\log{n} + n^2)$. Taking $l = n^2$ for the worst case, the complexity of the algorithm is of order $O(m\log^2(n) + n^2)$. It is better than the complexity order $O((m+n)n)$ of the algorithm given by Nguyen and Tran in~\cite{TFS2020} for the same problem. 

\section{Conclusions}
\label{sec: conc}

We have designed an efficient algorithm with the complexity $O((m\log{l} + n)\log{n})$ for computing the fuzzy partition corresponding to the greatest fuzzy auto-bisimulation of a finite fuzzy labeled graph $G$ under the G\"odel semantics, where $n$, $m$ and $l$ are the number of vertices, the number of non-zero edges and the number of different fuzzy degrees of edges of the input graph $G$, respectively. By using this algorithm, we have also provided an algorithm with the complexity $O(m\cdot\log{l}\cdot\log{n} + n^2)$ for computing the greatest fuzzy bisimulation between two finite fuzzy labeled graphs under the G\"odel semantics. This latter algorithm is better (has a lower complexity order) than the previously known algorithms for the considered problem. 

Our algorithms can be restated for other fuzzy graph-based structures such as fuzzy automata, fuzzy labeled transition systems, fuzzy Kripke models, fuzzy social networks and fuzzy interpretations in fuzzy description logics. 


\bibliography{BSfDL}
\bibliographystyle{elsarticle-harv}


\end{document}